\documentclass[journal]{IEEEtran}

\ifCLASSINFOpdf
\else
\fi

\usepackage{microtype}
\usepackage{graphicx}
\usepackage{subfigure}
\usepackage{booktabs,pifont} 

\usepackage{empheq}

\allowdisplaybreaks

\DeclarePairedDelimiter\abs{\lvert}{\rvert}
\usepackage{pgfplots}
\pgfplotsset{compat=1.12}
\usetikzlibrary{shapes,arrows}
\usetikzlibrary{positioning}
\usepackage{tikz}
\usetikzlibrary{positioning,chains,fit,shapes,calc}
\usepackage{amsmath,bm,times}
\usetikzlibrary{calc}
\usepackage{dsfont}
\usepackage{cuted}
\usepackage{caption}

\usepackage{mathtools}
\DeclarePairedDelimiter{\ceil}{\lceil}{\rceil}
\DeclarePairedDelimiter{\floor}{\lfloor}{\rfloor}
\usepackage{amsthm}
\usepackage{comment}
\usepackage{enumitem}
\usepackage{arydshln}
\usepackage{cite}
\usepackage{multirow}
\usepackage{calc}
\usepackage{blkarray}
\usepackage{url}
\theoremstyle{definition}
\newtheorem{theorem}{Theorem}

\newtheorem{lemma}{Lemma}
\newtheorem{claim}{Claim}
\newtheorem{proposition}{Proposition}
\newtheorem{corollary}{Corollary}
\newtheorem{example}{Example}
\newtheorem{remark}{Remark}
\newtheorem{definition}{Definition}
\usepackage{graphicx}
\usepackage{dblfloatfix}
\usetikzlibrary{decorations.pathreplacing}
\usepackage{blindtext, graphicx, amsfonts,
	amssymb,multirow,epstopdf}
\def\BibTeX{{\rm B\kern-.05em{\sc i\kern-.025em b}\kern-.08em
    T\kern-.1667em\lower.7ex\hbox{E}\kern-.125emX}}

\makeatletter
\renewcommand*\env@matrix[1][*\c@MaxMatrixCols c]{%
  \hskip -\arraycolsep
  \let\@ifnextchar\new@ifnextchar
  \array{#1}}
\makeatother
\usepackage[linesnumbered,ruled]{algorithm2e}

\newcommand{\cmark}{\ding{51}}%
\newcommand{\xmark}{\ding{55}}

\newcommand{\calM}{\mathcal{M}}
\newcommand{\calA}{\mathcal{A}}

\newcommand{\calD}{\mathcal{D}}
\newcommand{\calW}{\mathcal{W}}

\newcommand{\calX}{\mathcal{X}}
\newcommand{\calG}{\mathcal{G}}

\newcommand{\calU}{\mathcal{U}}
\newcommand{\calV}{\mathcal{V}}

\newcommand{\calN}{\mathcal{N}}
\newcommand{\calO}{\mathcal{O}}

\newcommand{\bfy}{\mathbf{y}}
\newcommand{\bfr}{\mathbf{r}}

\newcommand{\bfH}{\mathbf{H}}

\newcommand{\bfA}{\mathbf{A}}

\newcommand{\bfx}{\mathbf{x}}

\newcommand{\bfB}{\mathbf{B}}

\newcommand{\bfa}{\mathbf{a}}

\newcommand{\bfR}{\mathbf{R}}

\newcommand{\bfv}{\mathbf{v}}

\newcommand{\bfe}{\mathbf{e}}

\begin{document}

\title{Sparsity-Preserving Encodings for Straggler-Optimal Distributed Matrix Computations at the Edge}

\definecolor{mygr}{rgb}{0.6,0.4,0.0}
\definecolor{my1color}{rgb}{0.6,0.4,0.0}
\definecolor{mycolor1}{rgb}{0.00000,0.44700,0.74100}%
\definecolor{mycolor2}{rgb}{0.85000,0.32500,0.09800}%
\definecolor{mycolor3}{rgb}{0.45000,0.62500,0.19800}%
\tikzset{
block/.style    = {draw, thick, rectangle, minimum height = 2em, minimum width = 2em},
sum/.style      = {draw, circle, node distance = 1cm},
sum1/.style      = {draw, circle, minimum size = 1.1 cm},
input/.style    = {coordinate},
output/.style   = {coordinate},
}

\author{\IEEEauthorblockN{Anindya Bijoy Das, Aditya Ramamoorthy, David J. Love and Christopher G. Brinton} 
\thanks{Anindya Bijoy Das, David J. Love, and Christopher. G. Brinton are with the School of Electrical and Computer Engineering, Purdue University, West Lafayette, IN, USA 47907 (e-mail: \{das207, djlove, cgb\}@purdue.edu).}
\thanks{Aditya Ramamoorthy is with the Department of Electrical and Computer Engineering, Iowa State University, Ames, Iowa, USA 50011 (email: adityar@iastate.edu).}
\thanks{The material in this work has appeared in part \cite{10313473} at the 59th Annual Allerton Conference on Communication, Control, and Computing (Allerton), Monticello, IL, USA, 2023, [DOI: 10.1109/Allerton58177.2023.10313473].}
\thanks{This work was supported in part by the National Science Foundation (NSF) under Grants CNS-2212565, ITE-2326898 and CCF-2115200, and Office of Naval Research (ONR) under grant N000142112472.}
}

\IEEEtitleabstractindextext{%

\begin{abstract}
Matrix computations are a fundamental building-block of edge computing systems, with a major recent uptick in demand due to their use in AI/ML training and inference procedures. Existing approaches for distributing matrix computations involve allocating coded combinations of submatrices to worker nodes, to build resilience to slower nodes, called stragglers. In the edge learning context, however, these approaches will compromise sparsity properties that are often present in the original matrices found at the edge server. In this study, we consider the challenge of augmenting such approaches to preserve input sparsity when distributing the task across edge devices, thereby retaining the associated computational efficiency enhancements. First, we find a lower bound on the weight of coding, i.e., the number of submatrices to be combined to obtain coded submatrices, to provide the resilience to the maximum possible number of straggler devices (for given number of devices and their storage constraints). Next we propose distributed matrix computation schemes which meet the exact lower bound on the weight of the coding. Numerical experiments conducted in Amazon Web Services (AWS) validate our assertions regarding straggler mitigation and computation speed for sparse matrices.
\end{abstract}

}

\maketitle
\IEEEdisplaynontitleabstractindextext
\IEEEpeerreviewmaketitle

\vspace{-0.1 cm}
\section{Introduction}
\label{sec:intro}
Edge computing platforms are constantly struggling to keep pace with the escalating demands for data processing tasks. A vast volume of data is being generated at the network edge today due to the rise of the Internet of Things (IoT), which encompasses self-driving cars, drones, health monitoring devices, and many other intelligent applications \cite{9517801}. 
The mounting complexity of edge learning tasks, exemplified by the constantly growing sizes of deep neural networks and other AI/ML models, coupled with vast amounts of data used in training and inference, persistently obstruct scalability. 
As these models become more sophisticated, they necessitate increasingly robust computational resources and capabilities. 

A possible solution to this can be borrowed from distributed computing systems, where computationally intensive tasks are distributed across multiple worker nodes. In the edge computing context, an oversubscribed edge server can allocate a processing task across multiple edge/IoT devices that has access to. Specifically, a profoundly-complex learning task can undergo partitioning into numerous sub-learning tasks, which are subsequently dispatched to multiple edge devices for execution \cite{9563062}. Thus, the computational burden at the edge server can be parallelized, enhancing training speed.

In the edge learning context, particular attention must be paid to \textit{matrix computations}, which form the cornerstone of numerous data processing tasks in AI/ML. With the expansion of data sizes, these computations encompass high-dimensional matrices, leading to extended runtimes with all else held constant. Distributed edge computations will aim to mitigate this by segmenting the entire matrix operation into smaller submatrices, and dispersing them across multiple edge devices for parallel execution. Nonetheless, the overall execution time of a task can be markedly impacted by slower or failed edge devices, often referred to as ``stragglers'' \cite{ramamoorthyDTMag20}. 

Stragglers are prevalent across the edge due to different reasons including heterogeneous computation capabilities found across devices  \cite{ramamoorthyDTMag20, hosseinalipour2020federated} and network congestion \cite{10529949, 9174120}, and can pose challenges to completing edge learning tasks in a timely manner.
Since the system must wait for the slowest task to finish before proceeding, these are particularly problematic in environments where tasks are expected to complete synchronously. Scenarios affected by stragglers encompass various real-world applications in federated learning \cite{9252954, 9024521, das2023jsait_submitted}, blockchain \cite{9397781, 9170905, 10159020}, timely computing \cite{8849235, 9115646}, distributed optimization \cite{karakus2017straggler, 9248056, 8007058, 8262879},  smart healthcare \cite{9878261, 10018935} and other similar environments. Thus, mitigating the stragglers is essential for optimizing the synchronization and speed of task completion in such cases.


Several coding theory techniques \cite{lee2018speeding, das2019random, dutta2016short, yu2017polynomial, c3les, yu2020straggler, tandon2017gradient, dasunifiedtreatment,  9513242, das2023distributed, 8849468, 8919859} have been recently proposed to mitigate the effect of stragglers. A simple example \cite{lee2018speeding} for computing $\bfA^T \bfx$ using three  devices involves partitioning the matrix $\bfA$ into two block-columns, denoted as $\bfA = [\bfA_0 | \bfA_1]$. The edge devices are then assigned specific tasks: one computes $\bfA_0^T \bfx$, another computes $\bfA_1^T \bfx$, and the third computes $(\bfA_0 + \bfA_1)^T \bfx$. Each device then handles only half of the computational load, the system can recover $\bfA^T \bfx$ if any two out of the devices return their results. 
This implies that the system can withstand the failure or delay of a single straggler. In broader terms, the recovery threshold stands as a pivotal metric, denoting the minimum number of edge devices ($\tau$) needed to fulfill their tasks, thereby facilitating the recovery of $\bfA^T \bfx$ from any subset of $\tau$ devices. Note that this idea can be extended to matrix-matrix multiplication also.

While numerous existing works achieve the optimal recovery threshold \cite{yu2017polynomial, 8849468, 8919859, das2019random} for specific device and storage constraints, they are not exempt from limitations. Real-world datasets, pervasive across various fields like optimization, deep learning, power systems, and computational fluid dynamics, frequently manifest as sparse matrices. Leveraging this sparsity effectively holds the potential to significantly reduce the overall time required for matrix computations. However, techniques relying on MDS codes tend to generate dense linear combinations of submatrices \cite{yu2017polynomial, 8849468, 8919859, das2019random}, thereby obliterating the inherent sparsity within the matrix structure. 
Consequently, this can result in a severe reduction in the computational speed of edge devices \cite{wang2018coded}. In addition, this can lead to further delay, since the central server needs to transmit a larger number of non-zero entries to the edge devices. In this paper, our primary objective is to devise methodologies that integrate a relatively modest number of submatrices while preserving an optimal recovery threshold. This can enhance the computation speed and reduce the transmission delay, thereby improving the overall job completion speed.

\textbf{Organization:} In this work, first we formulate the problem, provide a literature background and summarize our contributions in Sec. \ref{sec:back}. Then, in Sec. \ref{sec:opt_weight}, we find a lower bound on the number of submatrices to be combined for coded submatrices that will provide resilience to the maximum number of stragglers in a system with given storage and computation constraints. After that, we develop  novel approaches for distributed matrix-vector multiplication (Sec. \ref{sec:prop_approach}) and distributed matrix-matrix multiplication (Sec. \ref{sec:matmat}), both of which meet that lower bound, maximizing sparsity preservation while providing resilience to the maximum number of stragglers. 
Next, we carry out experiments on an Amazon Web Services (AWS) cluster, and  provide the numerical results in Sec. \ref{sec:numexp}. Finally, Sec. \ref{sec:conclusion} concludes the paper with several possible future directions.

\section{Problem Formulation, Background and Summary of Contributions}
\label{sec:back}
In this section, we initially outline the problem of addressing sparsity in distributed matrix computations. Following that, we provide a concise overview of existing methods, analyze their limitations, and outline the key contributions of our research.
\vspace{-0.2 cm}

\subsection{Problem Formulation}
In this work, we investigate a distributed computing system comprised of a edge server and a set of $n$ edge devices. The goal is to compute $\mathbf{A}^T \mathbf{x}$ for matrix-vector multiplication or $\mathbf{A}^T \mathbf{B}$ for matrix-matrix multiplication. Here, $\mathbf{A} \in \mathbb{R}^{t \times r}$, $\mathbf{B} \in \mathbb{R}^{t \times w}$ are the ``input'' matrices, and $\mathbf{x} \in \mathbb{R}^{t}$ is a vector.

First, we consider a coded matrix-vector multiplication scheme where The primary objective of this system is to calculate the product $\bfA^T \bfx$, where $\bfA$ represents a sparse matrix and $\bfx$ denotes a vector. In line with previous approaches, we initially partition matrix $\bfA$ into $k_A$ distinct block-columns, $\bfA_0, \bfA_1, \bfA_2, \dots, \bfA_{k_A - 1}$; hence the goal is to recover $k_A$ corresponding unknowns $\bfA_0^T \bfx, \bfA_1^T \bfx, \bfA_2^T \bfx, \dots, \bfA_{k_A - 1}^T \bfx$.

Next, we consider the matrix-matrix multiplication case, where matrices $\bfA$ and $\bfB$ are partitioned into $k_A$ and $k_B$ disjoint block-columns, as $\bfA_0, \bfA_1, \bfA_2, \dots, \bfA_{k_A - 1}$ and $\bfB_0, \bfB_1, \bfB_2, \dots, \bfB_{k_B - 1}$. Thus, in this case we need to recover, in total, $k_A k_B$ unknowns in the form of $\bfA_i^T \bfB_j$ where $0 \leq i \leq k_A - 1$ and $0 \leq j \leq k_B - 1$.


The assumption is that the edge devices possess uniform memory capacity and computational speed. To elaborate, each worker can retain $\gamma_A = \frac{1}{k_A}$ portion of matrix $\bfA$, as well as the complete vector $\bfx$ (in the matrix-vector case), or $\gamma_B = \frac{1}{k_B}$ portion of matrix $\bfB$ (in the matrix-matrix case). In real-world scenarios, stragglers might emerge owing to discrepancies in computational speeds or instances of failure among designated edge devices at particular times \cite{das2019random}.

In this work, we allocate to each edge device a randomized linear combination of specific block-columns from matrix $\bfA$, along with the vector $\bfx$ in the matrix-vector case. For matrix-matrix multiplication, each edge device receives a randomized linear combination of certain block-columns from $\bfA$ and another combination from $\bfB$. However, as discussed in Section \ref{sec:intro}, assigning dense linear combinations risks losing the inherent sparsity of the matrices involved. To circumvent this, our aim is to distribute linear combinations involving fewer submatrices \cite{dasunifiedtreatment, das2023distributed}. To quantify this strategy, we introduce the notion of ``weight'', which plays a pivotal role in handling sparse matrices in distributed matrix computations.
\vspace{-0.05 cm}

\begin{definition}
\label{def:weight}
We define the ``weight'' $(\omega)$ for a coded matrix-computation scheme as the number of unknowns that participate within any matrix product computed by each of the edge device. Thus, if we combine $\omega_A$ submatrices of $\bfA$ to obtain the encoded submatrices for matrix-vector multiplication, we have $\omega = \omega_A$. In the matrix-matrix case, if we combine $\omega_A$ and $\omega_B$ submatrices of $\bfA$ and $\bfB$, respectively, to obtain the assigned encoded submatrices, then $\omega = \omega_A \omega_B$.   
\end{definition}

\vspace{-0.1 cm}
Thus, our goal is to obtain the optimal recovery threshold ($\tau = k_A$ for the matrix-vector case, and $\tau = k_A k_B$ for the matrix-matrix case \cite{yu2017polynomial}) while maintaining $\omega$ as low as possible. 

\vspace{-0.2 cm}

\subsection{Existing Methods and Our Motivations}
\label{sec:background}
Numerous coded computation schemes have been recently proposed for distributed matrix computations \cite{lee2018speeding, xhemrishi2022distributed, das2019random, dutta2016short, yu2017polynomial, das2020coded, hollanti2022secure, yu2020straggler, tandon2017gradient, aliasgari2020private, 8849468, 8919859, dasunifiedtreatment,  9513242, tandon2018secure, mallick2018rateless, 8849395}. 
In this section, we provide an overview of existing algorithms and examine their limitations, which have inspired our current research efforts. Note that there are methods in \cite{mallick2018rateless, 8849395, xhemrishi2022distributed, keshtkarjahromi2018dynamic} which are developed for matrix-vector multiplication only. In this work, in addition to the matrix-vector case, we also address the distributed matrix-matrix multiplication scenario, the more challenging one.

The initial studies in this domain focused on the recovery threshold as a pivotal metric. Unlike the methodologies outlined in \cite{wang2018coded, mallick2018rateless}, which exhibit suboptimal straggler resilience, the polynomial code approach emerges as one of the pioneering methods to achieve the optimal recovery threshold. We begin by illustrating a simplified example of this approach \cite{yu2017polynomial} in the context of distributed matrix-matrix multiplication.

Let us consider a system consisting of $n = 5$ edge devices, each capable of storing half of each of the matrices $\bfA$ and $\bfB$, hence $1/k_A = 1/k_B = 1/2$. We partition $\bfA$ and $\bfB$ into $k_A = k_B = 2$ block-columns each, denoted as $\bfA_0$, $\bfA_1$ and $\bfB_0$, $\bfB_1$, respectively. Defining matrix polynomials $\bfA(z) = \bfA_0 + \bfA_1 z$ and $\bfB(z) = \bfB_0 + \bfB_1 z^2$, we can write the product of these two matrix polynomials as $\bfA^T (z) \bfB(z) = \bfA_0^T \bfB_0 + \bfA_1^T \bfB_0 z + \bfA_0^T \bfB_1 z^2 + \bfA_1^T \bfB_1 z^3$.
The edge server assesses $\bfA(z)$ and $\bfB(z)$ at $n = 5$ distinct real numbers, and then, transmits the corresponding evaluated matrices to edge device $W_i$, where $0\leq i \leq n - 1$. Each device then computes its designated matrix-matrix block-product and sends back the result to the server. Given that $\bfA^T (z) \bfB(z)$ forms a degree-$3$ polynomial, upon receiving results from the fastest $\tau = 4$ devices, the server can decode all coefficients within $\bfA^T (z) \bfB(z)$, thereby obtaining $\bfA^T \bfB$ (since any $4\times 4$ submatrix of a $5\times 4$ Vandermonde matrix has a rank $4$). Hence, the recovery threshold is $\tau = 4$, indicating  resilience to $s = 1$ straggler.

For given storage constraints, i.e., storing $1/k_A$ and $1/k_B$ fractions of $\bfA$ and $\bfB$, respectively, if each device is responsible to compute $1/k$ fraction of the overall job (where $k = k_A k_B$), then $s = n - k_A k_B$ is the maximum number of stragglers that any scheme can be resilient to \cite{yu2017polynomial}. This can be an important property of a coded computation scheme, since resilience to a high number of stragglers is often crucial. For example, consider smart city IoT environments, where distributed learning is often incorporated within edge computing systems to provide intelligent, data-driven services \cite{9099242, 7973152}. Quality of Service (QoS) for these smart city use-cases have been shown to greatly improve through computation offloading, and thus highly depends on the system being resilient to slower devices.

While the polynomial code meets the lower bound on the recovery threshold, recent works on matrix computations have identified metrics beyond recovery threshold that also need to be considered. Table \ref{tab:compare} demonstrates an overall summary to compare available schemes
in terms of different other metrics. Here we discuss the importance of factoring them into our methodology which aims at improving those other metrics along with enjoying the optimal recovery threshold.

\begin{table*}[t]
\caption{{\small Comparison among existing works on coded matrix-computations (the approach in \cite{dutta2019optimal} involves a higher  computational complexity). However, the methods in \cite{mallick2018rateless, 8849395, xhemrishi2022distributed, keshtkarjahromi2018dynamic} are developed for matrix-vector multiplication only, and therefore, are not included in the comparison. Note that the method in \cite{das2023distributed} does not develop and meet the lower bound on the encoding weight, which is one of our contributions.}} 
\vspace{-0.2 cm}
\label{tab:compare}
\begin{center}
\begin{small}
\begin{sc}
\begin{tabular}{c c c c c}
\hline
\toprule
\multirow{2}{1 cm}{Codes} & Optimal & Numerical  & Sparsely & Heterogeneous \\
   & Threshold? & Stability? & Coded? & System?\\
 \midrule
Repetition Codes  & \xmark & \cmark & \cmark& \cmark\\ \hline
Product Codes \cite{lee2017high}, Factored Codes \cite{9513242}    &\xmark & \cmark   & \xmark& \xmark\\ \hline
Polynomial Codes \cite{yu2017polynomial}  & \cmark & \xmark  & \xmark& \cmark\\ \hline
MatDot Codes \cite{dutta2019optimal}  & \cmark & \xmark  & \xmark& \cmark\\ \hline
Orthogonal Polynomial  \cite{8849468}, RKRP code\cite{8919859},& {\cmark} & {\cmark}  & {\xmark}& {\cmark}\\ \hline
Bivariate  Polynomial Code \cite{9519610}  & \cmark & \xmark &  \xmark& \cmark\\ \hline
Convolutional Code \cite{das2019random}, Circular Rot. Matrix \cite{ramamoorthy2019numerically} & {\cmark} & {\cmark}  & {\xmark}& {\cmark}\\ \hline
$\beta$-level Coding \cite{das2020coded} & \xmark & \cmark  & \cmark& \xmark\\ \hline
Sparsely Coded Straggler Optimal Scheme \cite{das2020coded} & \cmark & \cmark  & \cmark& \xmark\\ \hline
Class-based Scheme \cite{dasunifiedtreatment} & \cmark & \cmark  & \cmark& \xmark\\ \hline
Cyclic Code (with random coefficients) \cite{das2023distributed} & \cmark & \cmark  & \cmark & \cmark\\ \hline
{\bf Proposed Scheme} & \cmark & \cmark  & \cmark & \cmark\\
\bottomrule
\end{tabular}
\end{sc}
\end{small}
\end{center}
\vspace{-0.5 cm}
\end{table*}%

\vspace{0.05 in}
{\bf Sparse Matrices:}
Sparse matrices are ubiquitous across various domains like optimization, deep learning, and electromagnetism, as evidenced by their prevalence in real-world datasets (see \cite{sparsematrices} for examples). Essentially, many practical scenarios involve matrices with sparse elements, offering a potential opportunity to reduce the edge device computation time significantly \cite{wang2018coded}. For example, for anomalous defect elimination in intelligent manufacturing, sparsity of the normal features can be very beneficial \cite{9887911, 9784867}. This can significantly enhance the speed of the overall training process, which often involves deep neural networks \cite{9475965, 9214506}.

Consider two column vectors, $\bfa$ and $\bfy$, both of length $m$, where $\bfa$ contains roughly $\psi m$ non-zero entries ($0 < \psi << 1$). Computing $\bfa^T \bfy$ consumes about $2 \psi m$ floating-point operations (FLOPs), in contrast to approximately $2m$ FLOPs which would be required for a dense vector $\bfa$. In a similar manner, the dense linear encoding strategies in  \cite{yu2017polynomial, 8849468, 8919859, das2019random} inflate the number of non-zero entries in encoded matrices, thus forfeit sparsity preservation. For a system with storage coefficients $1/k_A$ and $1/k_B$, the polynomial coding scheme \cite{yu2017polynomial} and its derivatives \cite{8849468} obtain the encoded submatrices of $\bfA$ and $\bfB$ by having linear combinations of $k_A$ and $k_B$ submatrices, respectively. Consequently, non-zero entries can increase by up to $k_A$ and $k_B$ times, respectively, compared to the original matrices, substantially elongating computation time. This underscores the necessity for schemes that minimize the fusion of uncoded submatrices.

Note that there are various methods in the literature \cite{wang2018coded, das2020coded, dasunifiedtreatment, xhemrishi2022distributed, 9965842} that highlight and leverage the inherent sparsity of the ``input'' matrices. However, they have other limitations. For example, the approach detailed in \cite{xhemrishi2022distributed} is only applicable to the matrix-vector case; does not apply to matrix-matrix multiplication. The assumptions in \cite{9965842} differ from ours as they involve the edge server in some matrix computations. Furthermore, the approach in \cite{wang2018coded} and the $\beta$-level coding scheme in \cite{das2020coded} do not meet the optimal recovery threshold, necessitating more devices to complete their tasks. 
On the other hand, the straggler optimal approaches in \cite{das2020coded, dasunifiedtreatment} assign multiple tasks to each edge device, some of which can be quite densely coded, leading to higher computation delay. The approach in \cite{das2023distributed} also addresses the sparsity issue, but it does not make any guarantees in terms of a theoretical lower bound on the encoding weights. This implies that an enhanced coding approach could further streamline computational speed beyond these techniques for sparse matrices. We provide theoretical guarantees for our proposed coding method in Sec. \ref{sec:opt_weight}.

\vspace{0.05 in}
{\bf Numerical Stability:} The numerical stability of the system stands as another crucial concern. Given that the encoding and decoding techniques in coded computation function within the real field, decoding the unknowns from a  system of equations may lead to highly inaccurate results if the corresponding system matrix is ill-conditioned. Round-off errors could magnify in the decoded outcome due to the elevated condition numbers of the decoding matrices. For instance, the polynomial code method outlined in \cite{yu2017polynomial} integrates Vandermonde matrices into the encoding process, known for their ill-conditioned nature. Literature addressing this challenge \cite{8849468, 8919859, das2019random} underscores the significance of minimizing the worst-case condition number $(\kappa_{worst})$ across different choices of stragglers.

However, several numerically stable methods rely on random codes \cite{8919859, das2019random, das2020coded, dasunifiedtreatment}, necessitating substantial time investment to find an optimal set of random coefficients to ensure the numerical stability of the system. This typically involves generating a set of random coefficients initially and assessing $\kappa_{worst}$ across all straggler permutations. This step iterates several times (e.g., $20$), retaining the set of coefficients yielding the minimum $\kappa_{worst}$. However, the latency introduced by this iterative process escalates with the number of edge devices, potentially causing delays in the encoding process. For example, the sparsely coded approaches in \cite{das2020coded, dasunifiedtreatment} involve significant delay for determining a ``good'' set of coefficients, as demonstrated numerically in Sec. \ref{sec:numexp}.

\subsection{Summary of Contributions}
\label{sec:summary}
\begin{itemize}
\item We address the straggler issue in distributed matrix computations for edge learning tasks, specifically focusing on sparse matrices. We define the concept of ``weight'' and determine its lower bound to ensure resilience against the maximum number of stragglers, maximizing the system's tolerance to delays or inefficiencies. 

\item Then, we develop an algorithm (Alg. \ref{Alg:New_matvec}) for distributed matrix-vector multiplication, which meets the exact lower bound to build resilience against maximum number of stragglers for any number of edge devices ($n$) and for any given storage constraint. 

\item Next, we develop an algorithm (Alg. \ref{Alg:New_matmat}) for straggler-resilient distributed matrix-matrix multiplication, which also meets the lower bound on the weight for different values of $n, k_A$ and $k_B$. Both of Alg. \ref{Alg:New_matvec} and Alg. \ref{Alg:New_matmat} can be extended to a heterogeneous setting where the edge devices can have different storage and computation abilities.

\item Subsequently, we analyze the per edge device computational complexity for our algorithms, and show that our approaches involve lower computational complexity than other recent approaches \cite{dasunifiedtreatment, das2023distributed}. In addition, we show that our approach also outperforms the approaches in \cite{das2020coded, dasunifiedtreatment} requiring less time to find the encoding coefficients for the numerical stability of the system.

\item Finally, we conduct extensive numerical experiments in the Amazon Web Services (AWS) cluster. Our results confirm the superiority of our approach to different dense coded approaches \cite{yu2017polynomial, 8849468, 8919859}, and to the recent approaches developed specifically for sparse matrices \cite{das2023distributed, dasunifiedtreatment}.
\end{itemize}

 Note that a preliminary version of this paper appeared in \cite{10313473}. Compared to the conference version, we have (i) developed Alg. \ref{Alg:New_matmat} for the more challenging case: matrix-matrix multiplication, (ii) proved necessary theorems for straggler-resilience, and (iii) conducted the corresponding numerical simulations.

\section{Minimum Weight of Encoding}
\label{sec:opt_weight}

We assume homogeneous weights, i.e., the same number ($\omega_A$) of uncoded $\bfA$ (and $\omega_B$ for $\bfB$) submatrices are combined to obtain the encoded $\bfA$ (and $\bfB$) submatrices. If we define $\omega$ as the number of participating unknowns in any computed results obtained from any edge device, in matrix-vector multiplication, we have $\omega = \omega_A$ and in the matrix-matrix case, $\omega = \omega_A \omega_B$. Now we state the following proposition which provides a lower bound on $\omega$ for any coded matrix computation scheme with resilience to $s = n - k$ stragglers, where $k = k_A$ for the matrix-vector case and $k = k_A k_B$ for the matrix-matrix case.

\begin{proposition}
\label{prop:lowerbound}
Consider a distributed computation system of $n$ total devices each of which can store $1/k_A$ fraction of matrix $\bfA$ (and $1/k_B$ fraction of matrix $\bfB$ for the matrix-matrix multiplication case). Now, assume that a coded matrix computation scheme aims at resilience to $s = n - k$ stragglers out of $n$ devices, where $k = k_A$ for the matrix-vector case and $k = k_A k_B$ for the matrix-matrix case. Any such scheme that partitions $\bfA$ into $k_A$ (and $\bfB$ into $k_B$) disjoint block-columns has to maintain a minimum homogeneous value for $\omega$, which is given by $\hat{\omega} = \lceil{\frac{(n-s)(s+1)}{n}}\rceil$.
\end{proposition}
\begin{proof}
Since the scheme aims at resilience to {\it any} $s$ stragglers, any scheme needs to ensure the presence of each of the corresponding $k$ unknowns in at least $s+1$ different devices. In other words, for the matrix-vector case, each of $k = k_A$ such $\bfA_i$'s (and for the matrix-matrix case, each of $k = k_A k_B$ such $\bfA_i^T \bfB_j$'s) has to participate within the encoded submatrices in at least $s+1$ different devices. Thus, the sum of the required number of appearances of all the uncoded unknowns is $k(s+1)$.

Now, we assume homogeneous $\omega$, i.e., each of these $n$ devices is assigned a linear combination of $\omega$ uncoded unknowns. Hence, the total number of appearances of all the uncoded unknowns over the set of the worker nodes is $n \omega$.

Therefore, to build the resilience to {\it any} $s$ stragglers, we need to have, $n \, \omega \geq k (s+1)$, hence, $\omega \geq {\frac{(n-s)(s+1)}{n}}$.Thus, the minimum weight $\hat{\omega} = \Bigl\lceil{\frac{(n-s)(s+1)}{n}}\Bigr\rceil$.
\end{proof}

Now we state a corollary which considers different values of $k_A$ in terms of $s$, and provides the corresponding optimal weights for coded sparse matrix-vector multiplication.

\begin{corollary} 
\label{cor:lowerbounds}
Consider the same setting as Prop. \ref{prop:lowerbound} for coded matrix computations. Now, we have the following cases.
\begin{itemize}
    \item (i) If $k > s^2$, then $\hat{\omega} = s + 1$.
    \item (ii) If $s \leq k \leq s^2$, then $\ceil{\frac{s+1}{2}} \leq \hat{\omega} \leq s$.
\end{itemize}
\end{corollary}

\begin{proof}
Since $n = k + s$, from Prop. \ref{prop:lowerbound}, we have 
\begin{align}
\label{eq:omega_bound}
    \hat{\omega} = \Bigl\lceil{\frac{k(s+1)}{k + s}}\Bigr\rceil = \Bigl\lceil{\frac{1 + s}{1 + \frac{s}{k}}}\Bigr\rceil ;
\end{align}hence, $\hat{\omega}$ is a non-decreasing function of $k$ for fixed $s$. 

\noindent {\bf Part (i):} When $k > s^2$, we have $\frac{s}{k} < \frac{1}{s}$, and $\frac{1 + s}{1 + \frac{s}{k_A}} > \frac{1 + s}{1 + \frac{1}{s}} = s$. Thus, from \eqref{eq:omega_bound}, $\hat{\omega} > s$. In addition, from \eqref{eq:omega_bound}, for any $s \geq 0$, we have $\hat{\omega} \leq s + 1$. Thus, we have $\hat{\omega} = s + 1$.

\noindent {\bf Part (ii):} If $k = s^2$, from \eqref{eq:omega_bound}, we have $\hat{\omega} = s$. Similarly, if $k = s$, from \eqref{eq:omega_bound}, we have $\hat{\omega} = \ceil{\frac{s+1}{2}}$. Thus, the non-decreasing property of $\hat{\omega}$ in terms of $k$ concludes the proof.
\end{proof}

Now we describe a motivating example below where the encoding scheme meets the lower bound mentioned in Prop. \ref{prop:lowerbound}.

\begin{example}
\label{ex:toy_matvec}
\input{matrixvector_n_6_new}
Consider a toy system for distributed matrix-vector multiplication with $n = 6$ edge devices each of which can store $1/4$ fraction of matrix $\bfA$. We partition matrix $\bfA$ into $k_A = 4$ disjoint block-columns, $\bfA_0, \bfA_1, \bfA_2, \bfA_3$. According to Prop. \ref{prop:lowerbound}, the optimal weight $\omega = \omega_A$ can be as low as $\Bigl\lceil\frac{k_A(s+1)}{k_A + s}\Bigr\rceil = 2$. Now, we observe that the way the jobs are assigned in Fig. \ref{matvec6_new} meets that lower bound, where random linear combinations of $\omega_A = 2$ submatrices are assigned to the devices. It can be verified that this system has a recovery threshold $\tau = k_A = 4$, and thus, it is resilient to any $s = 2$ stragglers.
\end{example}

\section{Proposed Matrix-vector Multiplication Approach}
\label{sec:prop_approach}
\vspace{-0.05 in}
In this section, we detail our overall approach for distributed matrix-vector multiplication which is outlined in Alg. \ref{Alg:New_matvec}. We partition matrix $\bfA$ into $k_A$ block columns, $\bfA_0, \bfA_1, \bfA_2, \dots, \bfA_{k_A - 1}$, and assign a random linear combination of $\omega_A$ (weight) submatrices of $\bfA$ to every edge device. We show that for given $n$ and $k_A$, our proposed approach provides resilience to maximum number of stragglers, $s = n - k_A$. In addition, our coding scheme maintains the minimum weight of coding as mentioned in Prop. \ref{prop:lowerbound}.

Formally, we set $\omega_A = \Bigl\lceil{\frac{k_A(s+1)}{k_A + s}}\Bigr\rceil$, and assign a linear combination of $\bfA_i, \bfA_{i + 1}, \dots, \bfA_{i+\omega_A - 1} \, \left(\textrm{indices modulo} \, k_A \right)$ to edge device $W_i$, for $i = 0, 1, \dots, k_A-1$, where the linear coefficients are chosen randomly from a continuous distribution. Next, we assign a random linear combination of $\bfA_{i\omega_A}, \bfA_{i\omega_A + 1}, \bfA_{i\omega_A + 2}, \dots, \bfA_{(i+1)\omega_A - 1} \, \left(\textrm{indices modulo} \, k_A \right)$ to edge device $W_{i}$, for $i = k_A, k_A + 1, \dots, n-1$. Note that every edge device also receives the vector $\bfx$. Once the fastest $\tau = k_A$ edge devices finish and return their computation results, the edge server decodes $\bfA^T \bfx$. Note that we assume $k_A \geq s$, i.e., at most {\it half} of the devices may be stragglers.

\begin{algorithm}[t]
	\caption{Proposed scheme for distributed matrix-vector multiplication}
	\label{Alg:New_matvec}
   \SetKwInOut{Input}{Input}
   \SetKwInOut{Output}{Output}
   \vspace{0.1 in}
   \Input{Matrix $\bfA$, vector $\bfx$, $n$-number of edge devices, $s$-number of stragglers, storage fraction $\gamma_A = \frac{1}{k_A}$, such that $k_A \geq s$.}
   Partition $\bfA$ into $k_A$ disjoint block-columns\;
   Set weight $\, \omega_A =\Bigl\lceil\frac{k_A(s+1)}{k_A + s}\Bigr\rceil$\;
   \For{$i\gets 0$ \KwTo $n-1$}{
   \eIf{$i < k_A$}
   {
   Define $T = \left\lbrace i, i+1, \dots, i + \omega_A - 1 \right\rbrace$ (reduced modulo $k_A$)\;
   }
   { 
   Define $T = \left\lbrace i \omega_A, i \omega_A + 1, \dots, (i+1)\omega_A - 1 \right\rbrace$ (reduced modulo $k_A$)\;
   }
   Create a random vector $\bfr$ of length $k_A$ with entries  $r_{m}$, $0\leq m \leq k_A - 1$\;
   Create a random linear combination of $\bfA_{q}$'s where $q \in T$, thus $\tilde{\bfA}_i = \sum\limits_{q \in T} r_{q} \bfA_q$\;
   Assign encoded submatrix $\tilde{\bfA}_i$ and the vector $\bfx$ to edge device $W_i$\;
   Edge device $W_i$ computes $\tilde{\bfA}_i^T \bfx$\;
   }
   \Output{The edge server recovers $\bfA^T \bfx$ from the returned results by the fastest $k_A$ devices.}
   \vspace{0.1 cm}
\end{algorithm}

\setlength{\textfloatsep}{0pt}

\vspace{-0.1 in}
\subsection{Straggler Resilience Guarantee}
\label{sec:matvecstr}
Next we state the Lemma \ref{lem:hall} which assists us to prove Theorem \ref{thm:matvec} that discusses straggler resilience of our scheme.

\begin{lemma}
\label{lem:hall}
Choose any $m \leq k_A$ edge devices out of all $n$ devices in the distributed system. Now, if we assign the jobs to the edge devices according to Alg. \ref{Alg:New_matvec}, the total number of participating uncoded $\bfA$ submatrices within those $m$ edge devices is lower bounded by $m$. 
\end{lemma}

\begin{proof}
First we partition all $n$ edge devices into {\it two} sets where the first set, $\calW_0$ includes the first $k_A$ devices and the second set, $\calW_1$, includes the next $s$ edge devices, i.e., we have 
\begin{align}
\label{eq:edgedevices}
    \calW_0 &= \left\lbrace W_0, W_1, W_2, \dots, W_{k_A - 1} \right\rbrace ; \nonumber \\
\textrm{and} \;\; \;  \calW_1 &= \left\lbrace W_{k_A}, W_{k_A+1}, \dots, W_{n-1}  \right\rbrace . 
\end{align} Thus, we have $|\calW_0| = k_A$ and $|\calW_1| = s \leq k_A$. Now, we choose any $m \leq k_A$ edge devices, where we choose $m_0$ devices from $\calW_0$ and $m_1$ devices from $\calW_1$, so that $m = m_0 + m_1$. We denote set of the participating uncoded $\bfA$ submatrices within those devices as $\calA_0$ and $\calA_1$, respectively. Hence, to prove the lemma, we need to show $|\calA_0 \cup \calA_1| \geq m$, for any $m \leq k_A$. 

First, according to Alg. \ref{Alg:New_matvec}, we assign a random linear combination of $\bfA_i, \bfA_{i + 1}, \bfA_{i + 2}, \dots, \bfA_{i+\omega_A - 1} \, \left(\textrm{indices modulo} \, k_A \right)$ to edge device $W_i \in \calW_0$. Thus, the participating submatrices are assigned in a cyclic fashion \cite{das2020coded}, and the total number of participating submatrices within any $m_0$ devices of $\calW_0$ is  
\begin{align}
\label{eq:m1}
|\calA_0| \geq \min (m_0 + \omega_A - 1, k_A).
\end{align} Next, we state the following claim for the number of participating submatrices in $\calW_1$; the detailed proof is in \cite{10313473}.

\begin{claim}
\label{clm:m1gw1}
Choose any $m_1 \geq \omega_A$ devices from $\calW_1$. The number of participating submatrices within these devices, $|\calA_1| = k_A$. 
\end{claim}

\noindent {\bf Case 1:} If $m_1 \leq \omega_A - 1$, from \eqref{eq:m1} we have 
\begin{align*}
    |\calA_0 \cup \calA_1| \geq |\calA_0| & =  \min (m_0 + \omega_A - 1, k_A) \;\; (\textrm{when $m_0 > 0$})\\ 
    & \geq  \min (m_0 + m_1, k_A) \geq m ,
\end{align*} since $m = m_0 + m_1 \leq k_A$. We take account the remaining scenario when $m_0 = 0$. In that case,
\begin{align*}
    |\calA_0 \cup \calA_1| \geq |\calA_1| \geq \omega_A \geq m_0 + m_1 = m.
\end{align*} Here, the inequality $|\calA_1| \geq \omega_A$ holds, because the number of unknowns participating in any device is $\omega_A$.

\noindent {\bf Case 2:} If $m_1 \geq \omega_A$, from Claim \ref{clm:m1gw1} we can say,
\begin{align*}
    |\calA_0 \cup \calA_1| \geq |\calA_1| =  k_A \geq m,
\end{align*} which concludes the proof of the lemma.
\end{proof}

\begin{example}
Consider the same scenario in Example \ref{ex:toy_matvec}, where $k_A = 4$ and $s = 2$, therefore, $\calW_0 = \{W_0, W_1, W_2, W_3\}$ and $\calW_1 = \{ W_4, W_5 \}$. Now, choose $m = 3$ devices, $W_0, W_1$ and $W_4$. Thus, $m_0 = 2$ and $m_1 = 1$. Now, from the figure, we have $\calA_0 = \{\bfA_0, \bfA_1, \bfA_2\}$ and $\calA_1 = \{\bfA_0, \bfA_1\}$. Hence, $|\calA_0 \cup \calA_1| = 3 \geq m$. Similar properties can be shown for any choice $m \leq k_A = 4$ different devices.
\end{example}

Now we state the following theorem which provides the guarantee of resilience to maximum number of stragglers for given storage constraints.

\begin{theorem}
\label{thm:matvec}
Assume that a system has $n$ edge devices each of which can store $1/k_A$ fraction of matrix $\bfA$ and the whole vector $\bfx$ for the distributed matrix-vector multiplication $\mathbf{A}^T \mathbf{x}$. If we assign the jobs according to Alg. \ref{Alg:New_matvec}, we achieve resilience to $s = n - k_A$ stragglers.
\end{theorem}

\begin{proof}
According to Alg. \ref{Alg:New_matvec}, first we partition matrix $\bfA$ into $k_A$ disjoint block-columns. Thus, to recover the matrix-vector product, $\bfA^T \bfx$, we need to decode all $k_A$ vector unknowns, $\bfA^T_0 \bfx, \bfA^T_1 \bfx, \bfA^T_2 \bfx, \dots, \bfA^T_{k_A - 1} \bfx$. We denote the set of these $k_A$ unknowns as $\calU$. Now we choose an arbitrary set of $k_A$ edge devices each of which corresponds to an equation in terms of $\omega_A$ of those $k_A$ unknowns. Denoting the set of $k_A$ equations as $\calV$, we can say,  $|\calU| = |\calV| = k_A$. 

Now we consider a bipartite graph $\calG = \calV \cup \calU$, where any vertex (equation) in $\calV$ is connected to some vertices (unknowns) in $\calU$ which participate in the corresponding equation. Thus, each vertex in $\calV$ has a neighborhood of cardinality $\omega_A$ in $\calU$. 

Our goal is to show that there exists a perfect matching among the vertices of $\calV$ and $\calU$. To do so, we consider $\bar{\calV} \subseteq \calV$, where $|\bar{\calV}| = m \leq k_A$. Now, we denote the neighbourhood of $\bar{\calV}$  as $\calN (\bar{\calV}) \subseteq \calU$. Thus, according to Lemma \ref{lem:hall}, for any $m \leq k_A$, we can say that $|\calN (\bar{\calV})| \geq m$. So, according to Hall's marriage theorem \cite{marshall1986combinatorial}, we can say that there exists a perfect matching among the vertices of $\calV$ and $\calU$.

Next we consider the largest matching where the vertex $v_i \in \calV$ is matched to the vertex $u_j \in \calU$, which indicates that $u_j$ participates in the equation corresponding to $v_i$. Now, considering $k_A$ equations and $k_A$ unknowns, we construct the $k_A \times k_A$ coding (or decoding) matrix $\bfH$ where row $i$ corresponds to the equation associated to $v_i$ where $u_j$ participates. We replace row $i$ of $\bfH$ by $\bfe_j$ where $\bfe_j$ is a unit row-vector of length $k_A$ with the $j$-th entry being $1$, and $0$ otherwise. Thus we have a $k_A \times k_A$ matrix where each row has only one non-zero entry which is $1$. In addition, since we have a perfect matching, $\bfH$ will have only one non-zero entry in every column. Thus, $\bfH$ is a permutation of the identity matrix, and therefore, $\bfH$ is full rank. Since the matrix is full rank for a choice of definite values, according to Schwartz-Zippel lemma \cite{schwartz1980fast}, the matrix continues to be full rank for random choices of non-zero entries. Thus, the edge server can recover all $k_A$ unknowns from any set of $k_A$ edge devices.
\end{proof}

\begin{example}
\label{ex:12_3}
\input{matrixvector_opt_n_12}
Consider a system with $n = 12$ devices each of which can store $1/9$-th fraction of matrix $\bfA$. We partition $\bfA$ as $\bfA_0, \bfA_1, \dots, \bfA_8$. According to Alg. \ref{Alg:New_matvec}, we set the weight $\omega_A = \Bigl\lceil\frac{k_A(s+1)}{k_A + s}\Bigr\rceil = 3$, and assign random linear combinations of $\omega_A$ submatrices to each device as shown in Fig. \ref{matvec12_opt}. It can be verified that $\bfA^T \bfx$ can be recovered from {\it any} $\tau = k_A = 9$ devices, therefore, the scheme is resilient to {\it any} $s = 3$ stragglers.
\end{example}

\begin{remark}
\label{rem:betterthanjsait}
Our proposed approach meets the lower bound on the weight for any $s \leq k_A$, as mentioned in Prop. \ref{prop:lowerbound}. On the other hand, the approach in \cite{das2023distributed} always assigns a weight $\min(s+1, k_A)$ which can be higher than ours when $s \leq k_A \leq s^2$ (e.g., Examples \ref{ex:toy_matvec} and \ref{ex:12_3}), and thus, may lead to reduction in worker computation speed compared to ours. 
\end{remark}

\subsection{Extension to Heterogeneous System}
\label{sec:hetero_mv}

In this section, we expand our approach described in Alg. \ref{Alg:New_matvec} to accommodate a heterogeneous system comprising $\bar{n}$ edge devices, each with varying computational abilities and communication speeds. The assumption is that the storage capacities and processing speeds of the edge devices (i.e., the overall system architecture) are known before job assignment. Similar to the approach in \cite{das2023distributed}, we assume that the system includes $\lambda$ different types of devices, indexed from $0$ to $\lambda - 1$. For simplicity, we sort the devices in non-ascending order based on their types.

Let $\alpha$ represent the number of assigned columns and $\beta$ the number of columns processed per unit time by the ``weakest'' type device \cite{das2023distributed}. In this setup, an edge device $W_i$ of type $j_i$ is allocated $c_{j_i} \alpha$ coded columns of the data matrix $\bfA$ and has a computation speed of $c_{j_i} \beta$, where $c_{j_i} \geq 1$ is an integer. A higher $c_{j_i}$ indicates a ``stronger'' device $W_i$, capable of storing and processing data $c_{j_i}$ times faster than the ``weakest'' type device.

Given that the devices are sorted in non-ascending order by type, we have $j_0 \geq j_1 \geq j_2 \geq \dots \geq j_{\bar{n}-1} = 0$, which ensures $c_{j_0} \geq c_{j_1} \geq c_{j_2} \geq \dots \geq c_{j_{\bar{n}-1}} = 1$. Note that when $\lambda = 1$ and all $c_{j_i} = 1$, the system reduces to the homogeneous case discussed in Section \ref{sec:matvecstr}, where $0 \leq i \leq n - 1$ and $j_i = 0$. 

Consider the ``weakest'' type of edge device, which requires $\mu$ units of time to process $\alpha$ columns of $\bfA$. Consequently, any edge device $W_i$ of type $j_i$ can process $c_{j_i} \alpha$ columns within the same time frame, $\mu$. From a computation and storage standpoint, this means that a node $W_i$ (type $j_i$) can be viewed as the equivalent of $c_{j_i} \geq 1$ of the ``weakest'' type edge devices. Thus, the $\bar{n}$ devices in our heterogeneous system can be conceptualized as a homogeneous system composed of $n = \sum_{i = 0}^{\bar{n} - 1} c_{j_i}$ ``weakest'' type edge devices.

In other words, a device $W_k$ in the heterogeneous system ($0 \leq k \leq \bar{n} - 1$) can be represented as a combination of devices $\bar{W}_m, \bar{W}_{m+1}, \dots, \bar{W}_{m + c_{kj} - 1}$ in a homogeneous system, where $m = \sum_{i = 0}^{k-1} c_{j_i}$, and $W_k$ is of type $j_i$. To further illustrate this, for any worker node index $\bar{k}_A$ (where $0 \leq \bar{k}_A \leq \bar{n} - 1$), we define $k_A = \sum_{i = 0}^{\bar{k}A - 1} c_{j_i}$ and $s = \sum_{i = \bar{k}A}^{\bar{n} - 1} c_{j_i}$, thus $n = \sum_{i = 0}^{\bar{n} - 1} c_{j_i} = k_A + s$.
Thus, a heterogeneous system of $\bar{n}$ edge devices can be thought as a homogeneous system of $n = k_A + s$ nodes, for any $\bar{k}_A$ ($0 \leq \bar{k}_A \leq \bar{n} - 1$). 

Next, similar to \cite{das2023distributed}, we state the following corollary (of Theorem \ref{thm:matvec}) that demonstrates the straggler resilience property of our proposed approach in a heterogeneous setting. The corresponding proof can be obtained based on results in \cite{das2023distributed}.

\begin{corollary}
 \label{cor:matvec}
Consider a heterogeneous system of $\bar{n}$ devices of different types and assume any $\bar{k}_A$ (where $0 \leq \bar{k}_A \leq \bar{n} - 1$). Now, if the jobs are assigned to the modified homogeneous system of $n = k_A + s$ ``weakest'' type edge devices according to Alg. \ref{Alg:New_matvec}, the system will be resilient to $s$ such nodes.
\end{corollary}


\begin{example}
\label{exmpl:hetero}
\input{hetero_8}
Consider the example in Fig. \ref{hetero8}, which involves $\bar{n} = 8$ edge devices. Suppose the computation capacities $c_{j_i}$ are as follows: $c_{j_i} = 3$ when $i = 0$, $c_{j_i} = 2$ when $i = 1,  2$, and $c_{j_i} = 1$ for $3 \leq i \leq 7$. Therefore, the total computation capacity is $n = \sum_{i = 0}^{\bar{n} - 1} c_{j_i} = 12$. Now, $\bar{k}_A = 5$, leading to $k_A = \sum_{i = 0}^{\bar{k}A - 1} c_{j_i} = 9$, and $s = \sum_{i = \bar{k}_A}^{\bar{n} - 1} c_{j_i} = 3$. Hence, each ``weakest'' device is assigned $\frac{1}{9}$ of the total job. This scheme is resilient to the failure of any $s = 3$ block-column processing units, meaning it can tolerate the failure of any three type $0$ devices or any one type $2$ device.

If the ``stronger'' devices (those with $c_{j_i} > 1$) fail to complete all their tasks on time, our proposed system can still utilize their partial computations. For instance, consider a scenario where none of $W_3, W_4, \dots, W_7$ is a straggler, but $W_0, W_1$ and $W_2$ are slower than their expected speed. If $W_0$ completes two out of its three assigned tasks, and each of $W_1$ and $W_2$ completes one out of their two assigned tasks, we can still recover the final result. This is because we only need to process $k_A = 9$ block-columns across all nodes. Thus, our approach can leverage partial stragglers ($W_0$, $W_1$, and $W_2$ in this example).

\end{example}

\begin{remark}
While we extend our method to the heterogeneous case in a similar manner as \cite{das2023distributed}, our scheme involves a smaller encoding weight ($\omega$) compared to \cite{das2023distributed}, which enhances the overall speed of the job for a sparse input matrix $\bfA$.
\end{remark}

\subsection{Computational Complexity for a edge device} 
\label{sec:compcomplexity}
In this work, we assume that the ``input'' matrix, $\bfA \in \mathbb{R}^{t \times r}$, is sparse, i.e., most of the entries of $\bfA$ are zero. Let us assume that the probability for any entry of $\bfA$ to be non-zero is $\mu$, where $\mu > 0$ is very small. According to Alg. \ref{Alg:New_matvec}, we combine $\omega_A$ submatrices (of size $t \times r/k_A$) to obtain the coded submatrices and assign them to the edge devices. Hence, the probability for any entry of any coded submatrix to be non-zero is $1 - (1 - \mu)^{\omega_A}$ which can be approximated by $\omega_A \mu$. Thus, in our approach, the per edge device computational complexity is $\calO \left( \omega_A \mu \times  \frac{rt}{k_A} \right)$ where $\omega_A = \Bigl\lceil\frac{k_A(s+1)}{k_A + s}\Bigr\rceil$.

On the other hand, the dense coded approaches \cite{yu2017polynomial,8849468, 8919859} combine $k_A$ submatrices for encoding, hence, their per edge device computational complexity is $\calO \left( k_A \mu \times  \frac{rt}{k_A} \right) = \calO \left( \mu \times  r t \right)$ which is $\frac{k_A}{\omega_A} \approx \frac{s + k_A}{ s + 1}$ times higher than that of ours. Moreover, the recent sparse matrix computations approach in \cite{das2023distributed} combines $s+1$ submatrices for encoding (when $s < k_A$). Thus, its corresponding computational complexity is $\calO \left( (s+1) \mu \times  \frac{rt}{k_A} \right)$; approximately $(1 + s/k)$ times higher than that of ours. We clarify this with the following example.

\begin{example}
Consider the same setting in Example \ref{ex:12_3} where $n = 12$, $k_A = 9$ and $s = 3$. In this scenario, the recent work \cite{das2023distributed} assigns random linear combinations of $\min(s+1, k_A) = 4$ submatrices to each device. Thus, our proposed approach enjoys a $25\%$ decrease in computational complexity, which could significantly enhance the overall computational speed.
\end{example}

\vspace{-0.02 in}
\subsection{Numerical Stability and Coefficient Determination Time}
\label{sec:trialtime}
In this section, we discuss the numerical stability of our proposed matrix-vector multiplication scheme. The condition number is widely regarded as a significant measure of numerical stability for such a system \cite{das2019random, 8849468, 8919859}. In the context of a system consisting of $n$ edge devices with $s$ stragglers, the worst-case condition number ($\kappa_{worst}$) is defined as the highest condition number among the decoding matrices when considering all possible choices of $s$ stragglers. In methods involving random coding like ours, the idea is to generate random coefficients multiple (e.g., 20) times  and selecting the set of coefficients that results in the lowest $\kappa_{worst}$.

In our proposed method, we partition matrix $\bfA$ into $k_A$ disjoint block-columns, which underscores the necessity to recover $k_A$ vector unknowns. Consequently, in each attempt, we must determine the condition numbers of ${n \choose k_A}$ decoding matrices, each of size $k_A \times k_A$. This whole process has a total complexity of $\calO\left( {n \choose k_A} k_A^3\right)$. On the other hand, the recent sparse matrix computation techniques, such as sparsely coded straggler (SCS) optimal scheme discussed in \cite{das2020coded} or the class-based scheme discussed in \cite{dasunifiedtreatment} partition matrix $\bfA$ into $\Delta_A = \textrm{LCM}(n, k_A)$ block-columns. Thus, in each attempt, they need to ascertain the condition numbers of ${n \choose k_A}$ matrices, each of which has a size $\Delta_A \times \Delta_A$, resulting in a total complexity of $\calO\left( {n \choose k_A} \Delta_A^3\right)$. Since $\Delta_A$ can be considerably larger than $k_A$, those methods involve significantly more complexity compared to our proposed scheme. 



\begin{remark}
We can also utilize our proposed approach in a D2D-enabled Federated Learning scenario as discussed in \cite{das2023jsait_submitted}. The assumption is that the edge devices may be responsible for generating local data and also for performing some computational tasks. The devices transmit their local submatrices to some other trusted devices, and build resilience against the computationally slower devices. Since our proposed approach involves a smaller weight, any device would require transmitting to less number of other devices compared to the approach in \cite{das2023jsait_submitted}, therefore, the overall process would be faster.
\end{remark}

\section{Proposed Matrix-matrix Multiplication Approach}
\label{sec:matmat}
In this section, we discuss our proposed distributed matrix-matrix multiplication approach. As we mentioned above, we consider a system of $n$ edge devices, each of which can store $1/k_A$ and $1/k_B$ fraction of matrices $\bfA$ and $\bfB$, respectively, and our aim is to build resilience to any $s = n - k$ stragglers, where $k = k_A k_B$. Similar to the matrix-vector case, we assume that $s \leq k$, i.e., not more than half of the devices will not be stragglers. The overall procedure is given in Alg. \ref{Alg:New_matmat}. 

Here we provide an overview of our proposed method. First we partition matrices $\bfA$ and $\bfB$ into $k_A$ and $k_B$ disjoint block columns, respectively, as $\bfA_0, \bfA_1, \bfA_2, \dots, \bfA_{k_A - 1}$ and $\bfB_0, \bfB_1, \bfB_2, \dots, \bfB_{k_B - 1}$, respectively. 
Then we consider the lower bound stated in Prop. \ref{prop:lowerbound}, and find $\hat{\omega} = \lceil{\frac{(n-s)(s+1)}{n}}\rceil$. Next, we find an $\omega \geq \hat{\omega}$ such that $\omega = \omega_A \omega_B$ where $1 < \omega_A < k_A $ and $1 < \omega_B < k_B $, $\omega_A | k_A$ and $\omega_B | k_B$. 

\begin{algorithm}[t]
	\caption{Proposed scheme for distributed matrix-matrix multiplication}
	\label{Alg:New_matmat}
   \SetKwInOut{Input}{Input}
   \SetKwInOut{Output}{Output}
   \Input{Matrices $\bfA$ and $\bfB$, $n$-number of edge devices, storage fraction $\gamma_A = \frac{1}{k_A}$ and $\gamma_A = \frac{1}{k_B}$ (w.l.o.g., $k_A \leq k_B$), and set $k = k_A k_B$; number of stragglers, $s \leq k$.}
   Partition matrices $\bfA$ and $\bfB$ into $k_A$ and $k_B$ block-columns, respectively\;
   Create a $n \times k_A$ random matrix $\bfR_A$ with entries  $r^A_{i,j}$, $0\leq i \leq n -1 $ and $0\leq j \leq k_A -1$\;
   Create a $n \times k_B$ random matrix $\bfR_B$ with entries  $r^B_{i,j}$, $0\leq i \leq n -1 $ and $0\leq j \leq k_B -1$\;
   Set weights $\omega_A | k_A$ and $\omega_B | k_B$ (where $\omega_A \leq \omega_B$) so that $\omega_A \omega_B \geq \Bigl\lceil\frac{k_A k_B(s+1)}{k_Ak_B + s}\Bigr\rceil$ with $1 < \omega_A < k_A $\;
   
   \For{$i\gets 0$ \KwTo $n-1$}{
   \eIf{$i < k_Ak_B$}
   {
   Define $T = \left\lbrace i, i+1, \dots, i + \omega_A - 1 \right\rbrace$ (entries reduced modulo $k_A$)\;
   Define $S = \left\lbrace j, j+1, \dots, j + \omega_B - 1 \right\rbrace$ (entries reduced modulo $k_B$) where $j = \floor{i/k_A}$\;
    }
   { 
   $T = \left\lbrace \ell \omega_A, \ell \omega_A+1, \dots, (\ell+1) \omega_A - 1 \right\rbrace$ (reduced modulo $k_A$) when $\ell = \textrm{mod} (i, k_A)$\;
    $S = \left\lbrace m \omega_B, m \omega_B+1, \dots, (m +1) \omega_B - 1 \right\rbrace$ (reduced modulo $k_B$) where $m = \Bigl\lfloor{ \frac{i \, \omega_A}{k_A}}\Bigr\rfloor$\;
   }
   
   Create a random linear combination of $\bfA_{q}$'s where $q \in T$, thus $\tilde{\bfA}_i = \sum\limits_{q \in T} r^A_{i,q} \bfA_q$\;
   Create a random linear combination of $\bfB_{q}$'s where $q \in S$, thus $\tilde{\bfB}_i = \sum\limits_{q \in S} r^B_{i,q} \bfB_q$\;
   The edge server assigns encoded submatrices $\tilde{\bfA}_i$ and $\tilde{\bfB}_i$ to edge device $W_i$\;
   Edge device $W_i$ computes $\tilde{\bfA}_i^T \tilde{\bfB}_i$\;
   }
   \Output{The edge server recovers $\bfA^T \bfB$ from the fastest $k_A k_B$ edge devices.}
   \vspace{0.1 cm}
\end{algorithm}
\setlength{\textfloatsep}{0pt}

Now, within the edge devices in $\calW_0 = \left\lbrace W_0, W_1, \dots, W_{k - 1} \right\rbrace$, we assign a random linear combination of $\bfA_i, \bfA_{i+1}, \dots, \bfA_{i+\omega_A - 1}$ (indices of $\bfA$ are reduced modulo $k_A$) to edge device $W_i$, $0 \leq i \leq k - 1$. In other words, the corresponding submatrices of $\bfA$ are shifted in a cyclic manner over the edge devices in $\calW_0$. Next we set $j = \floor{i/k_A}$, and assign $\bfB_j, \bfB_{j+1}, \dots, \bfB_{j+\omega_B - 1}$ (indices of $\bfB$ are reduced modulo $k_B$) to edge device $W_i$, where $0 \leq i \leq k - 1$. 

After that, within the edge devices in $\calW_1 = \left\lbrace W_{k}, W_{k+1}, \dots, W_{n-1}  \right\rbrace$, we assign a random linear combination of $\bfA_{\ell \omega_A}, \bfA_{\ell \omega_A+1}, \dots, \bfA_{(\ell+1) \omega_A - 1}$ (indices of $\bfA$ are reduced modulo $k_A$) to edge device $W_i$, where $\ell = \textrm{mod} (i, k_A)$ and $k \leq i \leq n - 1$. Next we set $m = \floor{\frac{i \omega_A}{k_A}}$, and assign $\bfB_{m \omega_B}, \bfB_{m \omega_B+1}, \dots, \bfB_{(m+1)\omega_B - 1}$ (indices of $\bfB$ are reduced modulo $k_B$) to edge device $W_i$, where $k \leq i \leq n - 1$. After assigning the jobs, all the edge devices will start computing their respective submatrix-products, and once the fastest $\tau = k_A k_B$ devices finish and return their results, the edge server recovers all the unknowns in the form of  $\bfA_i^T \bfB_j$, where $0 \leq i \leq k_A - 1$ and $0 \leq j \leq k_B - 1$.

\begin{example}
\input{matmat_30}
\label{ex:matmat}
Consider the example shown in Fig. \ref{fig:matmat20} where $n = 20, \gamma_A = \gamma_B = 1/4$. So, we partition $\bfA$ and $\bfB$ into $k_A = k_B = 4$ block-columns, respectively. In each of the devices, according to Alg. \ref{Alg:New_matmat}, we assign one coded submatrix from $\bfA$ and one from $\bfB$ which are linear combinations of $\omega_A = \omega_B = 2$ uncoded submatrices with coefficients chosen i.i.d. at random from a continuous distribution. It can be verified that this scheme is resilient to $s = n - k_A k_B = 4$ stragglers. 
\end{example}

\subsubsection{Structure of the Job Assignment}
To describe the structure of the proposed scheme, first we partition the edge devices in $\calW_0 = \left\lbrace W_0, W_1, W_2, \dots, W_{k - 1} \right\rbrace $ into $k_A$ disjoint classes, denoted by $\calM_i$'s, where any $\calM_i$ consists of all the edge devices $W_j$'s if $j \equiv i (\textrm{mod} \; k_A)$. In other words, $\calM_i = \left\lbrace W_i, \, W_{k_A+i}, \, W_{2 k_A+i}, \, \dots  \right\rbrace$, for $i = 0, 1, 2, \dots, k_A - 1$. Hence, $|\calM_i| = k_B$, for any $i$.

Moreover, according to our proposed scheme, the participating submatrices of $\bfA$ are the same over all the edge devices in any $\calM_i$. For instance, in Fig. \ref{fig:matmat20}, we have $\calM_0 =  \left\lbrace W_0, \, W_4, \, W_{8}, \, W_{12}   \right\rbrace$, and random linear combinations of $\bfA_0$ and $\bfA_1$ are assigned to all the corresponding edge devices. At this point, we define a set, $\calD^A_i = \left\lbrace \bfA_i, \bfA_{i+1}, \dots, \bfA_{i + \omega_A - 1} \right\rbrace$, which consists of the participating submatrices of $\bfA$ corresponding to edge device set $\calM_i$, where the indices are reduced modulo $k_A$. Now we state the following claim (with the proof in \cite{das2023distributed}) which gives a lower bound for the cardinality of the union of any arbitrary number of $\calD^A_i$'s.

\begin{claim}
\label{claim:cyclicA}
Consider any $q$ sets $\calD^A_i$'s, $q \leq k_A - \omega_A + 1$, denoted w.l.o.g.,  $\bar{\calD}^A_j$, $0 \leq j \leq q - 1$ arbitrarily. Then $\abs*{\bigcup\limits_{j = 0}^{q-1} \bar{\calD}^A_j} \geq \omega_A + q - 1$.
\end{claim}

Now, consider any particular $\calM_i$. According to Alg. \ref{Alg:New_matmat}, the participating submatrices of $\bfB$ are shifted in a cyclic fashion over the edge devices of any $\calM_i$. For instance, in Example \ref{ex:matmat}, the participating submatrices, $\bfB_0, \bfB_1, \bfB_2$ and $\bfB_3$, are shifted in a cyclic fashion within the edge devices of $\calM_0$, i.e., $W_0, W_4, W_{8}$ and $W_{12}$. Next, in the following claim, we find the minimum number of participating unknowns (in the form of  $\bfA_u^T \bfB_v$) within any $\delta$ edge devices from any $\calM_q$. 

\begin{claim}
\label{claim:cyclicB}
Consider $\calM_q$, $0 \leq q \leq k_A - 1$. Denote the minimum of total number of participating unknowns (in the form of  $\bfA_i^T \bfB_j$) within any $\delta$ edge devices from $\calM_q$ by $\rho$. 
Then 
\begin{align*}
    \rho = \omega_A \times \textrm{min} \left( \omega_B + \delta - 1, k_B\right) \, . 
\end{align*}
\end{claim}
\begin{proof}
The participating uncoded submatrices of $\bfB$ are shifted in a cyclic fashion within the edge devices of $\calM_q$. Thus according to the proof of cyclic scheme in Appendix C in \cite{das2020coded}, the minimum number of constituent submatrices of $\bfB$ within any $\delta$ edge devices of $\calM_q$ is $\textrm{min}\, (\omega_B + \delta - 1, k_B)$ if $\delta \geq 1$. Now the coded submatrices of $\bfB$ are multiplied by linear combinations of the same $\omega_A$ submatrices, thus the minimum of total number of participating unknowns (in the form of $\bfA^T_u \bfB_v$) is $\rho = \omega_A \times \textrm{min} \left( \omega_B + \delta - 1, k_B\right) $.
\end{proof}

\subsubsection{Rearrangement of ${\calM}_i$'s}
\label{subsec:rearrange}
Before stating the necessary theorem, we discuss a pre-processing step that rearranges the $\calM_i$'s. Choose any arbitrary $m_0$ edge devices ($m_0 \leq k_A k_B$), and assume that $\delta_i$ devices have been chosen from $\calM_i$, for $0 \leq i \leq k_A - 1$, so that $\sum_{i = 0}^{k_A - 1} \delta_i  = m_0$. Now, we 
rearrange the $\delta_i$'s in a non-increasing sequence so that $\tilde{\delta}_0 \geq \tilde{\delta}_1 \geq \tilde{\delta}_2 \geq \dots \geq \tilde{\delta}_{k_A - 1}$, and rename the corresponding $\calM_i$'s as $\tilde{\calM}_i$'s so that $\tilde{\delta}_i$ devices have been  chosen from $\tilde{\calM}_i$. 

Now, we denote $\rho_0$ as the minimum of total number of participating unknowns (in the form of  $\bfA_i^T \bfB_j$) within the $\tilde{\delta}_0 $ edge devices of $\tilde{\calM}_0$. Thus, according to Claim \ref{claim:cyclicB},
\begin{equation}
    \label{eq:rhozero}
    \rho_0 =  \omega_A \times \min ( \omega_B + \tilde{\delta}_0 - 1, k_B) \; .  
\end{equation}

After that, we move to $\tilde{\calM}_1, \tilde{\calM}_2, \dots, \tilde{\calM}_{k_A - \omega_A}$, sequentially, to find the number of additional participating unknowns within the corresponding $\tilde{\delta}_i $ edge devices of $\tilde{\calM}_i$, where $1 \leq i \leq k_A - 1$. We denote $\rho_i$ as the minimum number of such additional participating unknowns in $\tilde{\calM}_i$. 

Here, according to Claim \ref{claim:cyclicA}, $\abs*{\bigcup\limits_{j = 0}^{0} \bar{\calD}^A_j} \geq \omega_A$ and $\abs*{\bigcup\limits_{j = 0}^{1} \bar{\calD}^A_j} \geq \omega_A + 1$. Thus,  there will be at least one additional participating submatrix of $\bfA$ in $\tilde{\calM}_0 \cup \tilde{\calM}_1$ in comparison to $\tilde{\calM}_0$, and the property will continue to hold until we consider the set $\tilde{\calM}_0 \cup \tilde{\calM}_1 \cup \dots \cup \tilde{\calM}_0 \cup \tilde{\calM}_{k_A - \omega_A}$. Now, since the submatrices of $\bfB$ (which will be multiplied by the additional submatrix of $\bfA$) are shifted in a cyclic fashion within any $\tilde{\calM}_i$, 
\begin{align}
\label{eq:rhobigger}
    \rho_i = \textrm{min} \left( \omega_B + \tilde{\delta}_i - 1, k_B\right)  \; ,
\end{align}
for $1 \leq i \leq k_A - \omega_A$. Note that, $\rho_i$ has a trivial lower bound, {\it zero}, when $k_A - \omega_A + 1 \leq i \leq k_A - 1$. 

Now we state the following lemma that provides a lower bound on the minimum number of participating unknowns (in the form of $\bfA^T_u \bfB_v$) in the equations from any arbitrary $m_0 \leq k_A k_B$ devices. Note that, we assume $k_B \geq k_A$ without loss of generality; if $k_A > k_B$, we can compute $\bfA^T \bfB$ as $\left(\bfB^T \bfA\right)^T$ without any additional computational cost.

\begin{lemma}
\label{lem:no_of_unknowns}
For any arbitrary $k_A \geq 3$ and $k_B \geq 3$ (where $k_A \leq k_B$ without loss of generality), if we assign the jobs to $n = k + s$ edge devices (where $s \leq k = k_A k_B$) according to Alg. \ref{Alg:New_matmat}, then the minimum of total number of participating unknowns within any $m$ edge devices ($m \leq k$) will be lower bounded by $m$.
\end{lemma}

\begin{proof}
We prove the lemma in a similar manner as we have proved Lemma \ref{lem:hall}. First, we partition all $n$ edge devices into {\it two} sets where $\calW_0$ includes the first $k$ devices and $\calW_1$, includes the next $s$ edge devices, i.e., we have $\calW_0 = \left\lbrace W_0, W_1, W_2, \dots, W_{k - 1} \right\rbrace$ and $\calW_1 = \left\lbrace W_{k}, W_{k+1}, \dots, W_{n-1} \right\rbrace$, 
where $k = k_A k_B$. Thus, we have $|\calW_0| = k$ and $|\calW_1| = s \leq k$. Now, we choose any $m \leq k$ edge devices, where we choose $m_0$ devices from $\calW_0$ and $m_1$ devices from $\calW_1$, so that $m = m_0 + m_1$. 
We denote set of the participating unknowns within those devices as $\calX_0$ and $\calX_1$, respectively. Hence, to prove the lemma, we need to show $|\calX_0 \cup \calX_1| \geq m$, for any $m \leq k$.

Now we state the following claims to find the number of participating unknowns in $\calW_0$ and $\calW_1$, with the proofs in App. \ref{app:proofW0} and App. \ref{app:proofclaim2}, respectively..

\begin{claim}
\label{clm:m1gw0}
Choose any $m_0 \leq k$ devices from $\calW_0$ such that $m_0 > 0$. The number of participating submatrices within these devices is lower bounded by $\textrm{min} \, (m_0 + \omega - 1, k)$. 
\end{claim}
\begin{claim}
\label{clm:m1gw1_2}
Choose any $m_1 \geq \omega$ devices from $\calW_1$. The number of participating submatrices within these devices, $|\calX_1| = k$. 
\end{claim}

\noindent {\bf Case 1:} $m_1 \leq \omega - 1$. In this case, from Claim \ref{clm:m1gw0}, we have 
\begin{align*}
    |\calX_0 \cup \calX_1| \geq |\calX_0| & =  \min (m_0 + \omega - 1, k) \; \; (\textrm{when} \, \, m_0 > 0) \\ 
    & \geq  \min (m_0 + m_1, k) \geq m ,
\end{align*} since $m = m_0 + m_1 \leq k$. We take account the remaining scenario when $m_0 = 0$. In that case,
\begin{align*}
    |\calX_0 \cup \calX_1| \geq |\calX_1| \geq \omega \geq m_0 + m_1 = m.
\end{align*} Here, the inequality $|\calX_1| \geq \omega = \omega_A \omega_B$ holds, because the number of unknowns participating in any device is $\omega_A \omega_B$.

\noindent {\bf Case 2:} $m_1 \geq \omega$. In this case, from Claim \ref{clm:m1gw1_2} we can say,
\begin{align*}
    |\calX_0 \cup \calX_1| \geq |\calX_1| =  k \geq m,
\end{align*} which concludes the proof of the lemma.
\end{proof}

While Alg. \ref{Alg:New_matmat} and Lemma \ref{lem:no_of_unknowns} have been developed and proved for scenarios where $\omega_A | k_A$ and $\omega_B | k_B$, our approach is applicable for other values also. The job assignment in the edge devices can still be the same as mentioned in Alg. \ref{Alg:New_matmat}. 

\begin{example}
Consider a distributed system when $k > s^2$ as mentioned in Corollary \ref{cor:lowerbounds}(i). In this scenario, $\hat{\omega} = s + 1$. However, $|\calW_1| = s$, thus we have $m_1 \leq \omega - 1$. It indicates that only Case 1 is enough to conclude the proof of Lemma \ref{lem:no_of_unknowns} in this setting. Hence, we do not require the constraints $\omega_A | k_A$ and $\omega_B | k_B$ here, which were required only to prove Claim \ref{clm:m1gw1_2} (and hence, the corresponding Case 2).
\end{example}

Now we state the following theorem which provides the guarantee of resilience to maximum number of stragglers for given storage constraints.

\begin{theorem}
\label{thm:matmat}
Assume that a system has $n$ edge devices each of which can store $1/k_A$ and $1/k_B$ fractions of matrices $\bfA$ and $\bfB$ respectively, for distributed matrix-matrix multiplication $\mathbf{A}^T \mathbf{B}$; $k_A \leq k_B$. If we assign the jobs according to Alg. \ref{Alg:New_matmat}, we achieve resilience to $s = n - k_A k_B$ stragglers, where $s \leq k_A k_B$.
\end{theorem}

\begin{proof}
According to Alg. \ref{Alg:New_matmat}, first we partition matrix $\bfA$ and $\bfB$ into $k_A$ and $k_B$ disjoint block-columns, respectively. Thus, to recover the matrix-matrix product, $\bfA^T \bfB$, we need to decode all $k = k_A k_B$ matrix unknowns, in the form of $\bfA^T_i \bfB_j$, where $0 \leq i \leq k_A - 1$ and $0 \leq j \leq k_B - 1$. We denote the set of these $k$ unknowns as $\calU$. Now we choose an arbitrary set of $k_A k_B$ edge devices each of which corresponds to an equation in terms of $\omega = \omega_A \omega_B$ of those $k$ unknowns. Denoting the set of $k$ equations as $\calV$, we can say,  $|\calU| = |\calV| = k$. 

Now, similar to the proof of Theorem \ref{thm:matvec}, we consider a bipartite graph $\calG = \calV \cup \calU$, where any vertex (equation) in $\calV$ is connected to some vertices (unknowns) in $\calU$ which participate in the corresponding equation. Thus, each vertex in $\calV$ has a neighborhood of cardinality $\omega$ in $\calU$. Our goal is to show that there exists a perfect matching among the vertices of $\calV$ and $\calU$.  Now, according to Lemma \ref{lem:no_of_unknowns}, for any $m \leq k$, we can say that $|\calN (\bar{\calV})| \geq m$. 
Thus, similar to to the proof of Theorem \ref{thm:matvec}, according to Hall's marriage theorem \cite{marshall1986combinatorial} and Schwartz-Zippel lemma \cite{schwartz1980fast}, we can prove that the edge server can recover all $k = k_A k_B$ unknowns from any set of $k$ edge devices.
\end{proof}

\begin{remark}
While Theorem \ref{thm:matmat} has been developed considering a homogeneous setting of edge devices, it can be also extended to a heterogeneous setting \cite{das2023distributed}, similar to the matrix-vector case discussed in Sec. \ref{sec:hetero_mv}.
\end{remark}

\subsection{Computational Complexity for a Worker Node} 
\label{sec:compcomplexitymatmat}
In this work, we assume that $\bfA \in \mathbb{R}^{t \times r}$ and $\bfB \in \mathbb{R}^{t \times w}$ are sparse, only a small fraction ($\eta$) of the entries are non-zero. Therefore, in our case, the computational complexity for any edge device is $\calO \left( \omega_A \eta \times \omega_B \eta \times t \times \frac{r w}{k_A k_B} \right) = \calO \left(\omega_A \omega_B \eta^2 \times \frac{r w t}{k_A k_B} \right)$. On the other hand, any dense coded approach \cite{yu2017polynomial, 8849468} assigns linear combinations of $k_A$ and $k_B$ submatrices, hence, the corresponding computational complexity is approximately $\calO \left( k_A \eta \times k_B \eta \times t \times \frac{r w}{k_A k_B} \right) = \calO \left( \eta^2 \times r w t \right)$;  $\frac{k_A k_B}{\omega_A \omega_B}$ times larger than ours, as $\omega_A < k_A, \omega_B < k_B$.

\begin{example}
Consider the example in Fig. \ref{fig:matmat20}, where $k_A = k_B = 4$ and $s = 4$. The approach in \cite{das2023distributed} requires $\omega \geq 5$; hence sets $\omega_A = 3$ and $\omega_B = 2$. On the other hand, in our approach, we require $\omega = \big\lceil{\frac{16 \times 5}{20}}\big\rceil = 4$, hence, we set $\omega_A = \omega_B = 2$. This indicates a $33\%$ reduction of computational complexity in our case compared to the method in \cite{das2023distributed}.   
\end{example}

\vspace{-0.5 cm}
\subsection{Numerical Stability}
Similar to the discussion in the matrix-vector case in Sec. \ref{sec:trialtime}, we can develop a numerically stable scheme using Alg. \ref{Alg:New_matmat}. Note that in the matrix-matrix case, we need to find two ``good'' sets of random coefficients, one for the encoding of $\bfA$, and the other for the encoding of $\bfB$. We empirically demonstrate the numerical stability of our approach in Sec. \ref{sec:numexp}.

\section{Numerical Experiments}
\label{sec:numexp}


In this section, we conduct simulations on distributed matrix-computations, which can be incorporated within the framework of numerous data processing tasks. For example, high-dimensional linear transformations are vital for dimensionality reduction techniques such as principal component analysis (PCA) \cite{mackiewicz1993principal} and linear discriminant analysis (LDA) \cite{xanthopoulos2013linear}. 
Furthermore, they are fundamental to training deep neural networks and employing them for classification. For example, every layer of a fully-connected deep neural network necessitates matrix-matrix multiplications during both forward and backward propagation \cite{ramamoorthyDTMag20}.

For our numerical experiments, we compare the performance of our algorithm against various alternative techniques \cite{yu2017polynomial, 8849468, 8919859, das2020coded, dasunifiedtreatment, das2023distributed}. It is important to note the existence of several other methodologies tailored specifically for sparse matrix computations. Notably, the approach outlined in \cite{wang2018coded} lacks resilience to the maximum number of stragglers given certain storage constraints. Additionally, the strategy proposed in \cite{xhemrishi2022distributed}, which partitions edge devices into untrusted and partly trusted clusters, diverges from our foundational assumptions. The approach described in \cite{ji2022sparse}, which delegates certain tasks to the edge server to mitigate the probability of rank-deficiency during decoding, also does not align with our model assumptions. Therefore, these methodologies are excluded from consideration in our numerical experimentation section.

We explore a distributed system that consists of $n = 42$ edge devices with $s = 6$ stragglers. We assume that each device can store $\gamma_A = \gamma_B = \frac{1}{6}$ fraction of matrices $\bfA$ and $\bfB$. We consider sparse input matrices $\bfA$ of size $20,000 \times 15000$ and $\bfB$ of size $20,000 \times 12000$. 
We assume three different cases where the sparsity of $\bfA$ and $\bfB$ are $95\%$, $98\%$ and $99\%$, respectively, which indicate that randomly chosen $95\%$, $98\%$ and $99\%$ entries of the matrices are zero. It is worth noting that there exist numerous practical instances where data matrices demonstrate such (or, even more) levels of sparsity (refer to \cite{sparsematrices} for specific examples). We carry out the experiments on an AWS (Amazon Web Services) cluster, utilizing a {\tt c5.18xlarge} machine as the server and {\tt t2.small} machines as the edge devices.

\vspace{0.05 in}

{\bf Worker computation time:} Table \ref{table:worker_comp} presents a comparison among different methods based on the computation time required by edge devices to complete their respective tasks. In this scenario, where $k_A = k_B = 6$, the approaches described in \cite{yu2017polynomial, 8849468, 8919859} allocate two linear combinations of $k_A = k_B = 6$ submatrices of $\bfA$ and $\bfB$, respectively, to each of the edge devices. Consequently, the original sparsity of matrix $\bfA$ (and $\bfB$) is lost within the encoded submatrices. As a result, the edge devices experience a significantly increased processing time for their tasks compared to our proposed approach or the methods outlined in \cite{das2020coded, dasunifiedtreatment, das2023distributed}, which are specifically designed for sparse matrices and involve smaller weights.

To discuss the effectiveness of our approach in more details, we compare the weight of the coding of our approach against the approach in \cite{das2023distributed}. In this scenario, when $n = 42$ and $s = 6$, our approach sets the weight $\Bigl\lceil{\frac{(n-s)(s+1)}{n}}\Bigr\rceil = \Bigl\lceil{\frac{36 \times 7}{42}}\Bigr\rceil = 6$; thus we set $\omega_A = 2$ and $\omega_B = 3$. On the other hand, the approach in \cite{das2023jsait_submitted} uses a weight greater than or equal to $s + 1 = 7$, thus needs to set $\omega_A = 4$ and $\omega_B = 2$. Hence, our approach involves around $25\%$ less computational complexity per edge device, which is supported by the results in Table \ref{table:worker_comp}.

\begin{table*}[t]
\caption{{\small Comparison of worker computation time and communication delay for matrix-matrix multiplication for $n = 42, \gamma_A = \gamma_B = \frac{1}{6}$ when randomly chosen $95\%$, $98\%$ and $99\%$ entries of matrices $\bfA$ and $\bfB$ are zero.}}
\vspace{-0.3 cm}
\label{table:worker_comp}
\begin{center}
\begin{small}
\begin{sc}
\begin{tabular}{c c c c c c c c c}
\hline
\toprule
\multirow{2}{*}{Methods} & \multicolumn{3}{c}{Worker Comp. Time (in s)} & & \multicolumn{3}{c}{Communication Delay (in s)}  \\ \cline{2-4} \cline{6-8}
&  $\mu = 99\%$ &  $\mu= 98\%$ & $\mu = 95\%$ && $\mu = 99\%$ &  $\mu= 98\%$ & $\mu = 95\%$    \\
 \midrule
Poly. Code  \cite{yu2017polynomial} & $1.65$ & $5.19$  & $8.95$ && $0.78$ &  $1.42$ & $2.44$  \\
Ortho Poly Code \cite{8849468}    & $1.60$ & $5.16$  & $9.03$ && $0.80$ &  $1.48$ & $2.36$ \\
RKRP Code \cite{8919859}    & $1.61$ & $5.11$  & $8.96$ && $0.82$ &  $1.44$ & $2.39$  \\
SCS Optimal Scheme \cite{das2020coded}  & $1.09$ & $2.12$  & $5.23$ && $0.39$ &  $0.58$ & $0.93$ \\
Class-based Scheme \cite{dasunifiedtreatment} &$0.74$ & $1.21$  & $4.12$ && $0.30$ &  $0.47$ & $0.77$  \\
Cyclic Code \cite{das2023distributed}  & $0.76$ & $1.24$  & $4.19$ && $0.32$ & $0.51$ & $0.83$ \\
{\textbf{Proposed Scheme}}  & $\mathbf{0.61}$ & $\mathbf{1.04}$  & $\mathbf{3.11}$ && $\mathbf{0.25}$ & $\mathbf{0.39}$ & $\mathbf{0.65}$ \\
\bottomrule
\end{tabular}
\end{sc}
\end{small}
\end{center}
\vspace{-0.2 in}
\end{table*}%

{\bf Communication delay:} In Table \ref{table:worker_comp}, we also illustrate the delay incurred during the transmission of encoded submatrices from the server to the edge devices. The methods outlined in \cite{yu2017polynomial}, \cite{8849468}, and \cite{8919859} utilize dense linear combinations of submatrices, causing a notable rise in non-zero entries within these encoded submatrices. Consequently, transmitting this increased number of non-zero entries results in significant communication delays. In contrast, our proposed approach addresses this issue by employing encoded submatrices formed through linear combinations of a limited number of uncoded submatrices, thus markedly reducing communication delays.

In addition, here we compare our proposed approach against the method in \cite{das2023distributed} in terms of the approximate number of non-zero entries that need to be sent from the edge server to any edge device. We consider the case when the matrices $\bfA$ and $\bfB$ have approximately $99\%$ entries to be zero, hence, only $1\%$ entries are non-zero. Now, in this scenario ($n = 42, k_A = k_B = 6$ and $s = 6$), the corresponding number of non-zero entries for our proposed approach is approximately $\frac{20k \times 15k}{k_A}  \times 0.01 \times \omega_A + \frac{20k \times 12k}{k_B}  \times 0.01 \times \omega_B = 2.2 \times 10^6$, since we set $\omega_A = 2$ and $\omega_B = 3$.
On the other hand, the number of the non-zero entries to be sent to each edge device in the approach \cite{das2023distributed} is approximately $\frac{20k \times 15k}{k_A}  \times 0.01 \times \omega_A + \frac{20k \times 12k}{k_B}  \times 0.01 \times \omega_B = 2.8 \times 10^6$, since in that case $\omega_A = 4$ and $\omega_B = 2$.
Thus our proposed approach requires the server to transmit approximately $22\%$ less non-zero entries than the method in \cite{das2023distributed}; Table \ref{table:worker_comp} confirms this gain of our approach.

\vspace{0.04 in}

{\bf Numerical stability:} 
Next, we assess the numerical stability of distributed systems using different coded matrix computation techniques. We examine the condition numbers of the decoding matrices for various combinations of $n = 42$ workers and $s = 6$ stragglers. By comparing the worst-case condition number ($\kappa_{worst}$) across different methods, we present the $\kappa_{worst}$ values in Table \ref{table:kappa}. The polynomial coding method \cite{yu2017polynomial} encounters issues with ill-conditioned Vandermonde matrices, leading to pronounced numerical instability, as indicated by its notably elevated $\kappa_{worst}$ value. In contrast, our proposed technique, characterized by numerical robustness, yields a smaller $\kappa_{worst}$ compared to the method described in \cite{8849468} where the condition numbers grow exponentially in terms of $s = n - k$. Although the approach outlined in \cite{8919859} achieves a slightly reduced $\kappa_{worst}$ compared to ours, Table \ref{table:worker_comp} reveals substantial increases in computation and communication delays due to the assignment of dense linear combinations of submatrices to the edge devices.

\begin{table}[t]
\caption{\small Comparison among different approaches in terms of worst case condition number $\left(\kappa_{worst} \right)$ 
and the corresponding required time for $10$ trials to find a good set of random coefficients}
\vspace{-0.1 in}
\label{table:kappa}
\begin{center}
\begin{small}
\begin{sc}
\begin{tabular}{c c c}
\hline
\toprule
\multirow{2}{*}{Methods}  & $\kappa_{worst}$ for  & Req. time for \;\\
  & $n = 42$, $s = 6$ & $10$ trials\\

 \midrule

\; \; Poly. Code  \cite{yu2017polynomial}   & $3.67 \times 10^{23}$ & $0$ \\
\; \; Ortho-Poly\cite{8849468}    & $8.80 \times 10^{10}$ & $0$ \\
\; \; RKRP Code\cite{8919859}    & $4.84 \times 10^8$ & $84 \; \textrm{min}$ \\
\; \; SCS Opt. Sch. \cite{das2020coded}   & $7.88 \times 10^9$ & $16 \; \textrm{hr}$ \\
\; \; Class based \cite{dasunifiedtreatment} & $5.81\times 10^8$ & $15 \; \textrm{hr}$ \\
\; \; Cyclic Code  \cite{das2023distributed} & $2.13 \times 10^9$ & $85 
 \; \textrm{min}$ \\
\; \; {\textbf{Proposed Scheme}}  & $\mathbf{{7.95 \times 10^8}}$ & $\mathbf{{86 \; \textrm{min}}}$ \\
\bottomrule
\end{tabular}
\end{sc}
\end{small}
\end{center}
\end{table}%

\vspace{0.02 in}

{\bf Coefficient determination time:} 
Next, Table \ref{table:kappa} shows a comparative analysis of various methods with respect to the time required for performing 10 trials to obtain a ``good'' set of random coefficients that ensures numerical stability of the system. As explained in Section \ref{sec:trialtime}, the techniques proposed in \cite{das2020coded} and \cite{dasunifiedtreatment} involve partitioning matrix $\bfA$ into $\Delta_A = \textrm{LCM}(n, k_A)$ block-columns. For instance, when $n = 42$ and $s = 6$, $\Delta_A = 42$ is significantly larger than $k_A = 6$, which denotes the partition level in our approach. Consequently, when dealing with higher-sized matrices to determine the condition number, the methods proposed in  \cite{dasunifiedtreatment} and \cite{das2020coded} necessitate considerably more time compared to our approach.

\begin{remark}
We also conduct numerical simulations on matrix-vector multiplication (Sec. V in \cite{10313473}), yielding results that mirror a similar trend observed in the matrix-matrix scenario.
\end{remark}

{\bf Encoding weights:} We also compare our approach more closely against the recent method developed in \cite{das2023distributed}, in terms of encoding weights for different scenarios. The results are shown in Fig. \ref{fig:weight_comp}. First, we consider a system of $n = 30$ devices and $s = 9$ stragglers within the matrix-vector scenario. Our proposed approach always matches the theoretical lower bound on the weight developed in Prop. \ref{prop:lowerbound} in the matrix-vector case, and outperforms the approach in \cite{das2023distributed}.

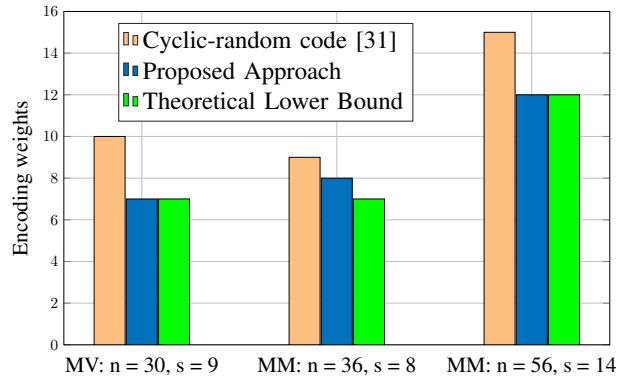
\begin{figure}[t]
\centering
\resizebox{0.93\linewidth}{!}{
\begin{tikzpicture}
\begin{axis}[
width = 5.5in,
height = 3.603in,
at={(2.6in,0.852in)},
major x tick style = transparent,
ybar=2*\pgflinewidth,
bar width=20pt,
ymajorgrids,
xmajorgrids,
xlabel style={font=\color{white!15!black}, font = \Large},
ylabel style={font=\color{white!15!black}, font = \Large},
ylabel={Encoding weights},
symbolic x coords={{\Large MV: n = 30, s = 9}, {\Large MM: n = 36, s = 8}, {\Large MM: n = 56, s = 14}},
xtick = data,
scaled y ticks = false,
enlarge x limits= 0.2,
ymin=0,
ymax=16,
log origin=infty,
legend cell align=left,
legend style={at={(0.1,0.67)}, nodes={scale=1.6}, anchor=south west, legend cell align=left, align=left, draw=white!15!black}
    ]
    
    \addplot[style={fill=orange!50,mark=none}]
            coordinates {({\Large MV: n = 30, s = 9}, 10) ({\Large MM: n = 36, s = 8}, 9) ({\Large MM: n = 56, s = 14}, 15)};
    \addlegendentry{Cyclic-random code \cite{das2023distributed}}
    
    \addplot[style={fill=mycolor1,mark=none}]
            coordinates {({\Large MV: n = 30, s = 9}, 7) ({\Large MM: n = 36, s = 8}, 8) ({\Large MM: n = 56, s = 14}, 12)};
    \addlegendentry{Proposed Approach}
    
    \addplot[style={fill=green,mark=none}]
             coordinates {({\Large MV: n = 30, s = 9}, 7) ({\Large MM: n = 36, s = 8}, 7) ({\Large MM: n = 56, s = 14}, 12) };
    \addlegendentry{Theoretical Lower Bound}
    \end{axis}

\end{tikzpicture}%
}
\caption{\small Comparison of encoding weights in matrix-vector (MV) and matrix-matrix (MM) multiplication between the method in \cite{das2023distributed}, our proposed approach, and the theoretical lower bound for different choices of $n$ and $s$.}
\vspace{0.1 in}
\label{fig:weight_comp}
\end{figure}

Next we consider the matrix-matrix scenario, for two different systems: (a) $n = 36$ and $s = 8$ and (b) $n = 56$ and $s = 14$. While our method involves slightly higher weights than the theoretical lower bound for system (a), it matches the bound for system (b), and outperforms the approach in \cite{das2023distributed} in both cases. Here, the theoretical lower bound in system (a) is $\big\lceil\frac{(36-8)(8+1)}{36}\big\rceil = 7$, which is a prime number. Thus, because of the divisibility issue with $\omega_A, \omega_B >1$, our proposed approach involves a slightly higher weight, $\omega = \omega_A \omega_B = 8$. Note that the dense coded approaches \cite{yu2017polynomial, 8849468, 8919859} involve a weight $n - s$, which is significantly higher than our proposed one and the method in \cite{das2023distributed}, and hence, they are not included in the comparison in Fig. \ref{fig:weight_comp}.

{\bf Numerical robustness and scalability}: Finally, we consider the case of matrix-vector multiplication in different system architectures having different number of edge devices ($n$) and stragglers ($s$). The $\kappa_{worst}$ values for different scenarios obtained from different available approaches are demonstrated in Fig. \ref{fig:kappaworst}. These results confirm that our proposed approach consistently leads to smaller $\kappa_{worst}$ values compared to the numerically stable approaches in \cite{8849468} and \cite{das2023distributed}, and provides competitive results compared to another numerically stable method in \cite{8919859} for even larger values of $n$ and $s$. This verifies the numerical robustness and the scalability of our proposed method for a larger network comprised of more edge devices. Note that in addition to the numerical robustness, unlike 
 \cite{8919859}, our approach can leverage the inherent sparsity of the ``input'' matrices, and thus leads to reduced computation and communication delays, as demonstrated in Table \ref{table:worker_comp}.
\begin{figure}[t]
\centering
\resizebox{0.9\linewidth}{!}{
\begin{tikzpicture}
\begin{axis}[
width = 5.5in,
height = 4.03in,
at={(2.6in,0.852in)},
major x tick style = transparent,
ybar=2*\pgflinewidth,
bar width=14pt,
ymajorgrids,
xmajorgrids,
xlabel style={font=\color{white!15!black}, font = \Large},
ylabel style={font=\color{white!15!black}, font = \Large},
ylabel={Worst case condition number ($\kappa_{worst}$)},
ymode=log,
symbolic x coords={{\Large n = 20, s = 3}, {\Large n = 20, s = 4}, {\Large n = 40, s = 3}, {\Large n = 40, s = 4}},
xtick = data,
scaled y ticks = false,
enlarge x limits= 0.15,
ymin=1e2,
ymax=1e10,
log origin=infty,
legend cell align=left,
legend style={at={(0.03,0.63)}, nodes={scale=1.6}, anchor=south west, legend cell align=left, align=left, draw=white!15!black}
    ]
    
    \addplot[style={fill=mycolor1,mark=none}]
            coordinates {({\Large n = 20, s = 3}, 3.8e4) ({\Large n = 20, s = 4}, 4.1e5) ({\Large n = 40, s = 3}, 1.8e6) ({\Large n = 40, s = 4}, 8.1e7)};
    \addlegendentry{Ortho-Poly Code \cite{8849468}}
    
    \addplot[style={fill=mycolor2,mark=none}]
            coordinates {({\Large n = 20, s = 3}, 9.1e3) ({\Large n = 20, s = 4}, 3.2e4) ({\Large n = 40, s = 3}, 3.9e5) ({\Large n = 40, s = 4}, 3.4e6)};
    \addlegendentry{RKRP Code \cite{8919859}}
    
    \addplot[style={fill=yellow!70,mark=none}]
            coordinates {({\Large n = 20, s = 3}, 7.9e4) ({\Large n = 20, s = 4}, 7.3e5) ({\Large n = 40, s = 3}, 1.3e7) ({\Large n = 40, s = 4}, 6.3e8)};
    \addlegendentry{Cyclic-random code \cite{das2023distributed}}
    
    \addplot[style={fill=green!70,mark=none}]
             coordinates {({\Large n = 20, s = 3}, 2.1e4) ({\Large n = 20, s = 4}, 6.9e4) ({\Large n = 40, s = 3}, 8.1e5) ({\Large n = 40, s = 4}, 1.7e7)};
    \addlegendentry{Proposed Approach}
    \end{axis}

\end{tikzpicture}%
}
\caption{\small Comparison of worst case condition numbers ($\kappa_{worst}$) of the encoding matrices obtained by different methods in distributed matrix-vector multiplication with different choices of $n$ and $s$.}
\vspace{0.1 in}
\label{fig:kappaworst}
\end{figure}
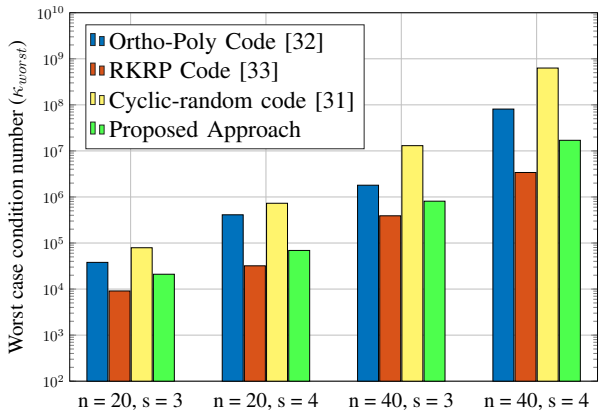

\section{Conclusion}
\label{sec:conclusion}

In this study, we developed novel distributed matrix computation techniques to maintain the sparsity characteristics of input matrices while ensuring resilience against a maximum number of stragglers within specified storage constraints. First, we find a lower bound on the homogeneous weight of the encoded submatrices, and we show that our approaches meet that lower bound. Unlike conventional dense coded methods \cite{yu2017polynomial, 8849468, 8919859, das2019random}, 
our proposed approach restricts coding within submatrices, thereby preserving the inherent sparsity of the input matrices $\bfA$ and $\bfB$ to a significant extent. Consequently, both worker computation and communication delays are markedly reduced compared to different coding techniques. In addition, unlike the sparsely coded methods in \cite{das2020coded, dasunifiedtreatment}, we demonstrate the extension of the approach to the heterogeneous setting, where the edge devices are rated with different computational and communication speeds.  

Future research avenues include exploring server-less architectures where edge devices collaborate directly, rather than relying on a edge server for matrix encoding \cite{9252954}. Enhancing performance in heterogeneous settings \cite{10024766} by allocating multiple jobs with varying weights as proposed in \cite{ozfatura2021coded, ozfatura2020age}, and devising schemes where knowledge about edge devices is limited prior to job assignment, can also be important directions for future investigation. 

\appendix

\subsection{Proof of Claim \ref{clm:m1gw0}}
\label{app:proofW0}
\begin{proof}
As discussed in Sec. \ref{subsec:rearrange},  we rearrange the ${\calM}_i$'s such that $\tilde{\delta}_0 \geq \tilde{\delta}_1 \geq \tilde{\delta}_2 \geq \dots \geq \tilde{\delta}_{k_A - 1}$ and rename the corresponding $\calM_i$'s as $\tilde{\calM}_i$'s so that $\tilde{\delta}_i$ devices have been  chosen from $\tilde{\calM}_i$. Next we denote the minimum number of participating unknowns in $\tilde{\calM}_i$ by $\rho_i$ when $0 \leq i \leq k_A - 1$. The trivial lower bound for $\rho_i$ is {\it zero} when $k_A - \omega_A + 1 \leq i \leq k_A - 1$, hence, $  \sum\limits_{i=0}^{k_A - 1} \rho_i \geq \sum\limits_{i=0}^{k_A - \omega_A} \rho_i$. Thus, in order to prove the lemma, for $m_0 = \sum\limits_{i = 0}^{k_A - 1} \tilde{\delta}_i$, we need to show that
\begin{align}
\label{eq:provehallfull}
    \sum\limits_{i=0}^{k_A - \omega_A} \rho_i \; \geq \; \textrm{min} \, (m_0 + \omega - 1, k) \; .
\end{align} 

Now we consider the following cases for each of which we show that \eqref{eq:provehallfull} is true. Note that the minimum value of $\tilde{\delta}_0 = 1$, since $m_0 \geq 1$.

{\it Case 1:} $1 \leq \tilde{\delta}_0 \leq k_B - \omega_B$. In this case, according to \eqref{eq:rhozero},
\begin{align}
\label{eq:case1}
    \rho_0 =  \omega_A \times \min ( \omega_B + \tilde{\delta}_0 - 1, k_B) = \omega_A (\omega_B + \tilde{\delta}_0 - 1).
\end{align}

{\it Case 1(a):} $\tilde{\delta}_{k_A - \omega_A + 1} = 0$. In this scenario, $\tilde{\delta}_{k_A - \omega_A + 1} = \dots = \tilde{\delta}_{k_A - 1} = 0$, hence $\sum_{i=k_A - \omega_A + 1}^{k_A - 1} \tilde{\delta}_i = 0$. Thus, 
\begin{align}
\label{eq:case1a}
    \sum\limits_{i=0}^{k_A - \omega_A} \rho_i \; & = \;     \rho_0 + \sum\limits_{i=1}^{k_A - \omega_A} \rho_i \; \nonumber \\ 
    & \geq  \omega_A (\omega_B + \tilde{\delta}_0 - 1) +  \sum\limits_{i=1}^{k_A - \omega_A} \tilde{\delta}_i + \sum\limits_{i=k_A - \omega_A + 1}^{k_A - 1} \tilde{\delta}_i \nonumber \\
    & =  \omega_A \tilde{\delta}_0 + \omega_A (\omega_B - 1) +  \sum\limits_{i=1}^{k_A - 1} \tilde{\delta}_i \nonumber \\
     & = (\omega_A - 1)\tilde{\delta}_0 + \tilde{\delta}_0 + \sum\limits_{i=1}^{k_A - 1} \tilde{\delta}_i +  \omega_A (\omega_B - 1) \nonumber \\
     & \geq \sum\limits_{i=0}^{k_A - 1} \tilde{\delta}_i + \omega_A \omega_B - 1 \, = \, m_0 + \omega - 1 \; ; 
\end{align} since $\tilde{\delta}_0 \geq 1$. Thus, \eqref{eq:provehallfull} is true in this scenario.

{\it Case 1(b):} $\tilde{\delta}_{k_A - \omega_A + 1} \geq 1$. In this scenario, $\tilde{\delta}_{k_A - \omega_A + 1} \geq \tilde{\delta}_{k_A - \omega_A + 2} \geq \dots \geq \tilde{\delta}_{k_A - 1}$, hence $\sum\limits_{i=k_A - \omega_A + 1}^{k_A - 1} \tilde{\delta}_i \leq (\omega_A - 1) \tilde{\delta}_0$. Next note that according to \eqref{eq:rhobigger}, we have $\rho_i \geq \omega_B + \tilde{\delta}_i  - 1$ for $1 \leq i \leq k_A - \omega_A$. Thus, we can say
\begin{align}
\label{eq:case1b}
    \sum\limits_{i=0}^{k_A - \omega_A} \rho_i \; & = \;     \rho_0 + \sum\limits_{i=1}^{k_A - \omega_A} \rho_i \; \nonumber \\ 
    & \geq  \omega_A (\omega_B + \tilde{\delta}_0 - 1) +  \sum\limits_{i=1}^{k_A - \omega_A} \left( \omega_B + \tilde{\delta}_i  - 1 \right) \nonumber \\
    & =  \tilde{\delta}_0 + (\omega_A - 1) \tilde{\delta}_0 + \omega_A (\omega_B - 1) \nonumber \\ & +  \sum\limits_{i=1}^{k_A - \omega_A} \tilde{\delta}_i + (k_A - \omega_A)(\omega_B - 1)\nonumber \\
    & \geq \sum\limits_{i=0}^{k_A - 1} \tilde{\delta}_i + (k_A - \omega_A)(\omega_B - 1) + \omega_A (\omega_B - 1) \nonumber \\
    & \geq \sum\limits_{i=0}^{k_A - 1} \tilde{\delta}_i + (\omega_B - 1) + \omega_A (\omega_B - 1) \nonumber \\
     & \geq \sum\limits_{i=0}^{k_A - 1} \tilde{\delta}_i + \omega_A \omega_B - 1 \, = \, m_0 + \omega - 1 \; ; 
\end{align} as we assume $\omega_A \leq \omega_B$. Thus, \eqref{eq:provehallfull} is true in this scenario too.

{\it Case 2:} $ k_B - \omega_B < \tilde{\delta}_0 \leq k_B$. In this case, according to \eqref{eq:rhozero}, we have $    \rho_0 =  \omega_A \times \min ( \omega_B + \tilde{\delta}_0 - 1, k_B) = \omega_A k_B$.

{\it Case 2(a):} $\tilde{\delta}_{k_A - \omega_A} > k_B - \omega_B$. In this scenario, $\tilde{\delta}_{0} \geq \tilde{\delta}_{1} \geq \dots \geq \tilde{\delta}_{k_A - \omega_A} = > k_B - \omega_B$. Thus, 
\begin{align}
\label{eq:case2a}
    \sum\limits_{i=0}^{k_A - \omega_A} \rho_i \; & = \;     \rho_0 + \sum\limits_{i=1}^{k_A - \omega_A} \rho_i \; \nonumber \\ & = \omega_A k_B + (k_A - \omega_A) k_B = k_A k_B = k.
\end{align} where $\rho_i = k_B$ according to \eqref{eq:rhobigger}, for $i = 1, 2, \dots, k_A - \omega_A$. Thus, \eqref{eq:provehallfull} is true in this scenario.

{\it Case 2(b):} $\tilde{\delta}_{k_A - \omega_A} \leq k_B - \omega_B$. Thus, according to \eqref{eq:rhobigger}, $\rho_{k_A - \omega_A} = \tilde{\delta}_{k_A - \omega_A} + \omega_B - 1$. Moreover, since in this scenario, we have $\tilde{\delta}_{k_A - 1} \leq \tilde{\delta}_{k_A -  2} \leq \dots \leq \tilde{\delta}_{k_A - \omega_A + 1} \leq k_B - \omega_B$; hence we have $\sum\limits_{i=k_A - \omega_A + 1}^{k_A - 1} \tilde{\delta}_i \leq (k_B - \omega_B) (\omega_A - 1)$. Now,
\begin{align}
\label{eq:case2b}
    \sum\limits_{i=0}^{k_A - \omega_A} \rho_i \; & = \;     \rho_0 + \sum\limits_{i=1}^{k_A - \omega_A - 1} \rho_i + \rho_{k_A - \omega_A}\; \nonumber \\ 
    & \geq  \omega_A k_B +  \sum\limits_{i=1}^{k_A - \omega_A - 1} \tilde{\delta}_i + \left( \tilde{\delta}_{k_A - \omega_A} + \omega_B - 1 \right)\nonumber \\
    & =  (\omega_A - 1)k_B + k_B + \omega_B - 1 +  \sum\limits_{i=1}^{k_A -\omega_A} \tilde{\delta}_i \nonumber \\
    & \geq  (\omega_A - 1)k_B + \omega_B - 1 - (k_B - \omega_B) (\omega_A - 1) \nonumber \\ 
    & + \tilde{\delta}_0 + \sum\limits_{i=1}^{k_A -\omega_A} \tilde{\delta}_i + \sum\limits_{i=k_A - \omega_A + 1}^{k_A -1} \tilde{\delta}_i \nonumber \\
    & =  \omega_A k_B - k_B + \omega_B - 1 - \omega_A k_B + \omega_A \omega_B \nonumber \\ 
    & + k_B - \omega_B  + \sum\limits_{i=1}^{k_A - 1} \tilde{\delta}_i \nonumber \\
     & \geq \sum\limits_{i=1}^{k_A - 1} \tilde{\delta}_i + \omega_A \omega_B - 1 \; . 
\end{align} Thus, \eqref{eq:provehallfull} is true in this scenario; this concludes the proof.

\end{proof}

\subsection{Proof of Claim \ref{clm:m1gw1_2}}
\label{app:proofclaim2}
\begin{proof}
Consider the edge devices in $\calW_1$. According to Alg. \ref{Alg:New_matmat}, we assign a linear combination of $\bfA_{\ell \omega_A}, \bfA_{\ell \omega_A + 1},  \dots, \bfA_{(\ell+1)\omega_A - 1} \, \left(\textrm{indices modulo} \; k_A \right)$ to edge device $W_{i}$, for $i = k, k + 1, \dots, n-1$ where $\ell = \textrm{mod}(i, k_A)$. 
In addition, we assign another linear combination of $\bfB_{m \omega_B}, \bfB_{m \omega_B + 1}, \dots, \bfB_{(m+1)\omega_B - 1} \, \left(\textrm{indices mod} \; k_B \right)$ to device $W_{i}$, for $i = k, k + 1, \dots, n-1$, where $m = \Bigl\lfloor{ \frac{i \, \omega_A}{k_A}}\Bigr\rfloor$.

Thus, the participating unknowns in edge device $W_{k_A}$ are $\bfA^T_0 \bfB_0, \dots, \bfA^T_0 \bfB_{\omega_B - 1}, \bfA_1^T \bfB_0, \dots, \dots,  \bfA_{\omega_A - 1}^T \bfB_{\omega_B - 1}$. 
Similarly, the participating submatrices in $W_{k_A + 1}$ are $\bfA^T_{\omega_A} \bfB_0, \dots, \bfA^T_{\omega_A} \bfB_{\omega_B - 1}, \bfA_{\omega_A+1}^T \bfB_0, \dots, \dots,  \bfA_{2\omega_A - 1}^T \bfB_{\omega_B - 1}$. 
In a consequence, $\omega = \omega_A \omega_B$ number of submatrices participate in each of those $s$ edge devices in a cyclic fashion. 

Now, denote the number of appearances of any submatrix $\bfA_i^T \bfB_j$ within the devices in $\calW_1$ by $\bfv_{ij} \geq 0$. Thus, for any $0 \leq i, p \leq k_A - 1$ and $0 \leq j, q \leq k_B - 1$, we have $ |\bfv_{ij}  -  \bfv_{pq}| \leq 1$, where $\sum_{i = 0}^{k_A - 1} \sum_{j = 0}^{k_B - 1}\bfv_{ij} = s \omega$. Thus, the average of these $\bfv_{ij}$'s is $\kappa = \frac{s \omega}{k}$. 
Since for every pair of $(i,j)$ and $(p,q)$, we have $|\bfv_{ij}  -  \bfv_{pq}| \leq 1$, thus, we can say that $\bfv_{ij} \geq \floor{\kappa} = \bigg\lfloor{\frac{s \omega}{k}}\bigg\rfloor$.
It indicates that within all $s$ devices of $\calW_1$, every unknown participates in at least $\floor{\kappa}$ times over $\floor{\kappa}$ distinct devices. In other words, any unknown may not participate in at most $s - \floor{\kappa}$ devices within the devices of $\calW_1$.

First, consider the case, $k = s$. Here, every unknown participates in $\floor{\kappa} = \omega$ devices, therefore, any unknown does not participate in $s - \omega$ devices. But, we choose any $m_1 \geq \omega$ devices in $\calW_1$, where $\omega = \ceil{\frac{s+1}{2}}$, since $k = s$. Thus,
\begin{align*}
   2 \omega \geq s + 1 > s \;\; \textrm{which indicates that}, \;\; \omega > s - \omega.
\end{align*}In addition, since $m_1 \geq \omega$, we claim that $m_1 > s - \omega_A$. Thus, every submatrix will participate at least once within those chosen $m_1$ devices, hence $|\calX_1| = k = k_A k_B$. 

Next, consider the other case when  $k > s$. Again, since we choose any arbitrary $m_1 \geq \omega$ devices in $\calW_1$, we are leaving $s - m_1$ devices in $\calW_1$. But 
\begin{align*}
    s - m_1 \leq s - \omega < s - \floor{\kappa}.
\end{align*} The second inequality holds since $s < k$, hence, $\floor{\kappa} < \omega$. Thus, every submatrix will participate at least once within those $m_1 \geq \omega$ devices, hence $|\calX_1| = k$.
\end{proof}

\ifCLASSOPTIONcaptionsoff
  \newpage
\fi

\bibliographystyle{IEEEtran}
\bibliography{citations}
\end{document}